\documentclass[a4paper, 12pt, oneside]{article}
\pdfoutput=1
\usepackage[a4paper, top=25.4mm, bottom=25.4mm, left=25.4mm, right=25.4mm]{geometry}
\usepackage{setspace}
\usepackage[T1]{fontenc}
\usepackage{ae,aecompl}
\usepackage[utf8]{inputenc}
\usepackage[english]{babel}
\usepackage[pdftex,  
            , pdfauthor={P. Coretto and F. Giordano},
            , pdftitle={NONPARAMETRIC ESTIMATION OF THE DYNAMIC RANGE OF MUSIC SIGNALS},
            , pdfsubject={Scientific Article},
            , pdfkeywords={Random subsampling, nonparametric regression, music data, dynamic range},
            , pdfproducer={},
            , pdfcreator={pdfTeX - Emacs (Linux) }
            , colorlinks=true, citecolor=blue, linkcolor=blue, urlcolor=blue
            , breaklinks=true
            , bookmarksnumbered=true,bookmarksopen=true,bookmarksopenlevel=1,
            , pdftex=true]{hyperref}
\usepackage{etoolbox}
\usepackage{xcolor}
\usepackage{url}
\usepackage{natbib} 
\usepackage{graphicx}
\DeclareGraphicsExtensions{.pdf,.png,.jpg, .eps}
\usepackage{caption}
\usepackage{multirow}
\usepackage{amssymb}
\usepackage{subfig}
\usepackage{amsmath, amsthm, amssymb, latexsym}
\usepackage{float}
\usepackage{booktabs}
\usepackage{datetime}
\theoremstyle{plain}

\newtheorem{proposition}{Proposition}
\newtheorem{lemma}{Lemma}
\newtheorem{corollary}{Corollary}
\newtheorem{assump}{}
\newenvironment{assumption}[2]{
  \renewcommand\theassump{#1}
  \begin{assump}[#2]}{\end{assump}}
\theoremstyle{definition}

\newcommand{\set}[1]{\left\{#1\right\}}
\newcommand{\tonde}[1]{ \left( #1 \right)}
\newcommand{\quadre}[1]{ \left[ #1 \right]}
\newcommand{\eps}{\varepsilon}%
\newcommand{\abs}[1]{\left| #1 \right| }

\newcommand{\uno}{\text{\bf \large  1}}
\newcommand{\Inf}{\infty}

\newcommand{\Z}{\mathbb{Z}}
\DeclareMathOperator{\ex}{E}   %
\DeclareMathOperator{\var}{Var}%
\newcommand{\conv}{\rightarrow}
\newcommand{\convd}{\stackrel{\text{d}}{\longrightarrow}}
\newcommand{\convp}{\stackrel{\text{p}}{\longrightarrow}}

\newcommand{\ProcEps}{ \{\varepsilon_t\}_{t \in \mathbb{Z}}}

\usepackage[ruled,vlined]{algorithm2e}

\SetKwInOut{Data}{data}
\SetKwInOut{Parameter}{parameter}
\SetKwInOut{Input}{input}
\SetKwInOut{Output}{output}

\usepackage{mdframed}
\newmdenv[
  backgroundcolor=gray!20,
  frametitle=,
  skipabove=\topsep,
  skipbelow=\topsep,
]{reminder}

\date{}
\title{\Large  NONPARAMETRIC ESTIMATION OF THE \\ DYNAMIC RANGE OF MUSIC SIGNALS%
\footnote{We thank Bob Katz and  Earl Vickers (both  members of the Audio Engineering Society) for the precious feedback on some of the idea contained in the paper. The authors gratefully acknowledge support from the University of Salerno grant program ``Finanziamento di attrezzature scientifiche e di supporto per i Dipartimenti e i Centri Interdipartimentali - prot. ASSA098434, 2009''.}}
\author{Pietro Coretto \\ Universit\`{a} di Salerno, Italy \\ pcoretto@unisa.it
\and %
Francesco Giordano \\ Universit\`{a} di  Salerno, Italy \\ giordano@unisa.it }

\begin{document}
\begin{reminder}
{\centering \color{red} \textbf{\textsf{This is a preprint. The revised version of this paper is published as}}\\}
\vspace{5pt}
P. Coretto and F. Giordano (2017). %
``Nonparametric estimation of the dynamic range of music signals''. %
\textit{Australian \& New Zealand Journal of Statistics}, Vol. 59(4), pp. 389--412 %
(\href{http://doi.wiley.com/10.1111/anzs.12217}{\sf download link}).
\end{reminder}

{\let\newpage\relax\maketitle}

\maketitle


\begin{quote}\small{
{\bf Abstract.~}\noindent  The dynamic range is an important parameter which measures the spread of sound power, and for music signals it is a measure of recording quality. There are various descriptive measures of sound power, none of which has strong statistical foundations. We start from a nonparametric model for sound waves where an additive stochastic term has the role to catch transient energy.  This component is recovered by a simple rate-optimal kernel estimator that requires a single data-driven tuning.  The distribution of its variance is approximated by a consistent random subsampling method that is able to cope with the massive size of the typical dataset.  Based on the latter, we propose a statistic, and an estimation method that is able to represent the dynamic range concept consistently. The behavior of the statistic is assessed based on a large numerical experiment where  we simulate dynamic compression on a selection of real music signals. Application of the method to real data also shows how the proposed method can predict subjective experts' opinions about the hifi quality of a recording.


\noindent {\bf AMS2010:} 97M80 (primary); 62M10, 62G08, 62G09, 60G35 (secondary)
}\end{quote}

\section{Introduction}\label{Section_Introduction}

Music signals have a fascinating complex  structure with interesting statistical
properties.  A music signal is the sum of periodic components plus transient
components that determine changes from one dynamic level to another. The term
``transients'' refers to  changes in acoustic energy. Transients  are of huge
interest. For technical reasons most recording and listening medium have to
somehow compress  acoustic energy variations, and this causes that  peaks are
strongly reduced with respect to  the average  level. The latter is also known
as ``dynamic range compression''. Compression of the dynamic range (DR)
increases the perceived loudness. The DR of a signal is the spread of acoustic
power. Loss of DR along the recording-to-playback chain translates into  a loss
of audio fidelity. While DR is a well  established technical concept, there is
no consensus on how to define it and how to measure it, at least in the field of
music signals. DR measurement has become a hot topic in the audio business.  In
2008 the release of the album ``Death Magnetic'' (by Metallica),  attracted
medias' attention for its extreme and aggressive loud sounding approach caused
by massive DR compression. DR manipulations are not reversible, once applied the
original dynamic is lost forever \citep[see][and references therein]{katz-2007,
  vickers-2010}. Furthermore, there is now  consensus that there is a strong
correlation between the DR and the recording quality perceived by the
listener. Practitioners in the audio industry use to measure the  DR  based  on
various descriptive  statistics for which little is known in terms of their
statistical  properties \citep{boley-etal-2010, ballou-2005}.

The aim of this paper is twofold: (i) define a DR statistic that is able to
characterize the dynamic of a music signal, and to detect DR compression
effectively, (ii) build a procedure to estimate DR with proven  statistical
properties.  In signal processing a  ``dynamic compressor'' is a device that
reduces  the peakness of the sound energy. The idea here is that dynamic
structure of a music signal is characterized by the energy produced by transient
dynamic, so that  the DR is measured by looking at the distribution of transient
power. We propose a nonparametric model composed by two elements: (i)  a smooth
regression function mainly accounting for long term harmonic components; (ii) a
stochastic component representing transients. In this framework transient power
is given by the variance of the stochastic component. By consistently estimating
the distribution of the variance of the stochastic component, we obtain  the
distribution of its power  which, in turn, is the  basis for constructing our DR
statistic. The DR,  as well as other background concepts are given in Sections
\ref{sec:tech} and \ref{sec:Statistical_Modelling_and_Estimation}.  

This paper gives four contributions. The idea of decomposing the music signal
into a deterministic function of time plus a stochastic component is due to the
work of \cite{Serra_Smith_1990}. However, it is usually assumed that stochastic
term of this decomposition is white noise. While this is appropriate in some
situations, in general the white noise assumption is too restrictive.  The first
novelty in this paper is that we propose a decomposition where the stochastic
term   is an  $\alpha$-mixing process, and this  assumption allows to
accommodate transient structures beyond those allowed by linear processes. The
stochastic component is obtained by filtering out the smooth component of the
signal, and  this is approximated with a simple kernel estimator inspired to
\cite{altman-1990}. The second contribution of this work is that we develop upon
Altman's seminal paper obtaining a rate optimal kernel estimator without
assuming linearity and knowledge of the correlation structure of the stochastic
term.  An important advantage of the proposed smoothing is that only one
data-driven tuning is needed (see  Assumption \ref{ass:M} and Proposition
\ref{prop:1}), while existing methods require two tunings to be fixed by the
user \citep[e.g. ][]{Hall-a-etal-1995}.  Approximation of the distribution of
the variance of the stochastic component of the signal is done by a subsampling
scheme inspired to that developed by \cite{politis-romano-1994} and
\cite{politis-romano-wolf-2001}. However, the standard  subsampling requires to
compute the variance of the stochastic component on the entire sample, which in
turn means that we need to compute the kernel estimate of the smooth component
over  the entire sample. The latter is unfeasible given the astronomically large
nature of the typical sample size. Hence, a third contribution of the paper is
that  we propose a consistent random subsampling scheme that does only require
computations at  subsample level.   The smoothing and the subsampling are
discussed in Section \ref{sec:estimation}.  A further contribution of the paper
is that  we  propose a DR statistic based on the quantiles of the variance
distribution of the stochastic component. The smoothing--subsampling previously
described is used to obtain estimates of such a statistic. The performance of
the DR statistic is assessed in a  simulation study where we use real data to
produce simulated levels of compression. Various combinations of compression
parameters are considered. We show that the proposed method is quite accurate in
capturing the DR concept. DR is considered as a measure of hifi quality, and
based on a real dataset we show how the estimated DR measure  emulates
comparative subjective judgements about hifi quality given by experts. All this
is treated in Sections \ref{sec:empirical_evidence_montecarlo} and
\ref{sec_real_dataset}.  Conclusions and final remarks are given in Section
\ref{sec:concl-final-remarks}. All proofs of statements are given in the final
Appendix.

\section{Background concepts: sound waves, power and dynamic}\label{sec:tech}

Let $x(t)$ be a continuous time waveform taking values on the time interval $T_1
\leq t \leq T_2,$ such that $\int_{T_1}^{T_2} x(t) dt=0$. Its power is given by 
\begin{equation}\label{eq:power}
P^{RMS}=C\sqrt{\frac{1}{T_2 - T_1} \int_{T_1}^{T_2} x(t)^2 dt}
\end{equation}
where $C$ is an appropriate scaling constant that depends on the measurement
unit. \eqref{eq:power} defines the so-called root mean square (RMS) power.  It
tells us that the power expressed by a waveform is determined by the average
magnitude of the wave swings around its average level. In other words the
equation \eqref{eq:power} reminds us of the concept of standard deviation.  
A sound wave  $x(t)$ can be recorded and stored by means of analog and/or
digital processes.  In the digital world $x(t)$ is represented numerically by
sampling and quantizing the analog version of $x(t)$.  The sampling scheme
underlying the so-called Compact Disc Digital Audio, is called Pulse Code
Modulation (PCM).  In PCM sampling a voltage signal $x(t)$ is sampled as a
sequence of integer values proportional to the level of $x(t)$ at equally spaced
times $t_0, t_1,\ldots$~.  The CDDA is based on PCM with sampling frequency
equal to 44.1KHz, and 16bits precision. The quantization process introduces
rounding errors also known as quantization noise. Based on the PCM samples
$\{x_0,x_1, \ldots x_T\}$, and under strong conditions on the structure of the
underlying $x(t)$, the RMS power can be approximated  by 
\begin{equation}\label{eq:pcm_power}
P^{RMS}_T = \sqrt{\frac{1}{T+1} \sum_{t=0}^{T} x_t^2}.
\end{equation}
The latter is equal to sampling variance, because this signals have zero
mean. Power encoded in a PCM stream is expressed as full-scale decibels:  
\begin{equation}\label{eq:pcm_dbfs}
\text{dBFS} = 20\log_{10}\frac{P^{RMS}_T}{P_0},
\end{equation}
where $P_0$ is the RMS power of a reference wave. Usually  $P_0 = 1/\sqrt{2}$,
that is, the RMS power of a pure sine-wave, or  $P_0=1$ that is, the RMS power of
pure square-wave. For simplicity we set $P_0=1$ in this paper. dBFS is commonly
considered as DR measure because it measures the spread between sound power and
power of a reference signal. 

For most real-world signal power changes strongly over time. In Figure
\ref{fig:intheflash2039} we report a piece of sound extracted from the left
channel of the song ``In the Flesh?'' by Pink Floyd. The song starts with a soft
sweet lullaby corresponding to Block 1 magnified in the bottom plot. However, at
circa 20.39s the band abruptly starts  a sequence of blasting riffs. This teaches 
us that: (i) sound power of  music signals can change tremendously over time;
(ii)  the power  depends on $T$, that is the time horizon.  In the audio
engineering community the practical approach is to time-window the signal and
compute average power across windows,  then several forms of DR statistics are
computed \cite[see][]{ballou-2005} based on dBFS. In practice one chooses a $T$,
then splits the PCM sequence into blocks of $T$ samples allowing a certain
number $n_o$ of overlapping samples between blocks, let $\overline{P^{RMS}}_T$
be the average of $P_T^{RMS}$ values computed on each block, finally a simple DR
measure, that we call ``sequential DR'',  is computed as 
\begin{equation}\label{eq:AvgDR}
\text{DRs} = - 20 \log_{10} \frac{\overline{P^{RMS}}_T}{x_{\text{peak}}},
\end{equation}
where $x_{\text{peak}}$ is the peak sample. The role of  $x_{\text{peak}}$ is to
scale the DR measure so that it does not depend on the quantization
range. DRs=10  means that on average the RMS power is 10dBFS below the maximum
signal amplitude. Notice that 3dBFS increment translates into twice the power.
DRs numbers are easily interpretable, however, the statistical foundation is
weak. Since the blocks are sequential, $T$ and  $n_o$ determine the blocks
uniquely  regardless the structure of the signal at hand. The second issue is
whether the average  $\overline{P^{RMS}}_T$ is a good  summary of the power
distribution in order to express the DR concept. Certainly the descriptive
nature of the DRs  statistic, and the lack of a stochastic framework, does not
allow to make inference and judge numbers consistently. 
\begin{figure}[!t]
\centering 
\includegraphics[width=\textwidth]{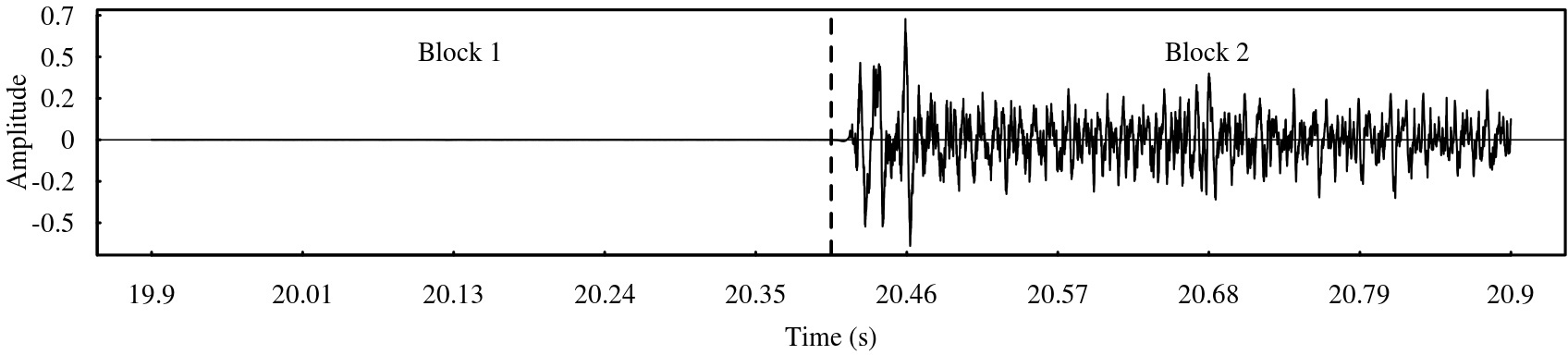}\\
\vspace{5mm}
\includegraphics[width=\textwidth]{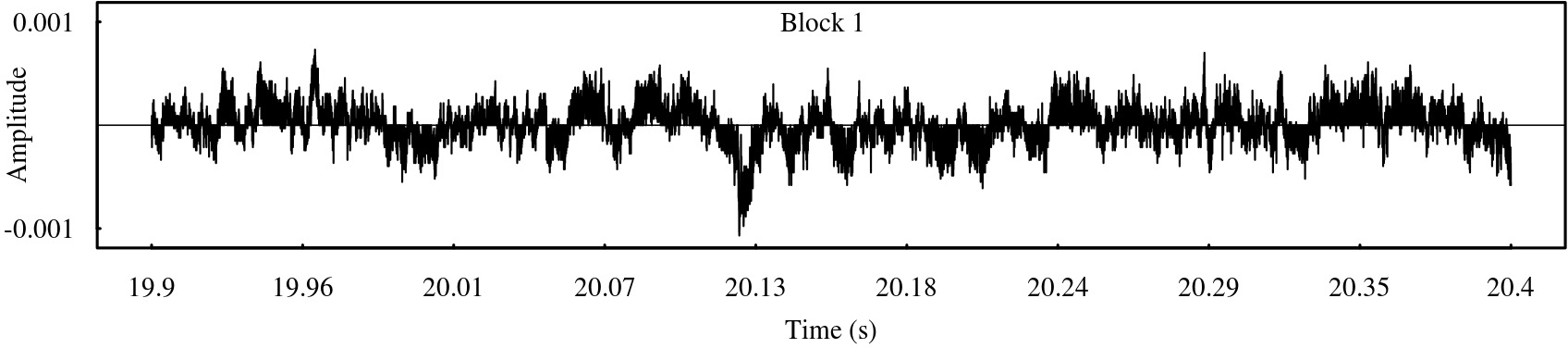}
\caption{Waveform extracted from the  left channel of the song ``{\sl In the Flesh?}'' 
from ``{\sl The Wall''} album  by Pink Floyd (Mobile Fidelity Sound Lab remaster, 
catalog no. UDCD 2-537). Top plot  captures 980ms of music centered at 
time 20.39524s (vertical dashed line).  Bottom plot magnifies  Block 1.}%
\label{fig:intheflash2039}
\end{figure}


\section{Statistical properties of sound waves and modelling}\label{sec:Statistical_Modelling_and_Estimation}

The German theoretical physicist Herman Von Helmholtz
(\citeyear{vonhelmholtz-1885})  discovered that within small time intervals
sound signals produced by instruments are periodic and hence representable  as
sums of periodic functions of time also known as ``harmonic components''. The
latter implies a discrete power spectrum. \cite{risset-mathews-1969} discovered
that the intensity of the harmonic components varied strongly over time even for
short time lengths. \cite{Serra_Smith_1990} proposed to model sounds from single
instruments by a sum of sinusoids plus a white noise process. While the latter
can model simple signals, e.g. a flute playing a single tune, in general such a
model is too simple to represent more complex sounds. 

\begin{figure}[t]
\begin{center}
\includegraphics[width=0.9\textwidth]{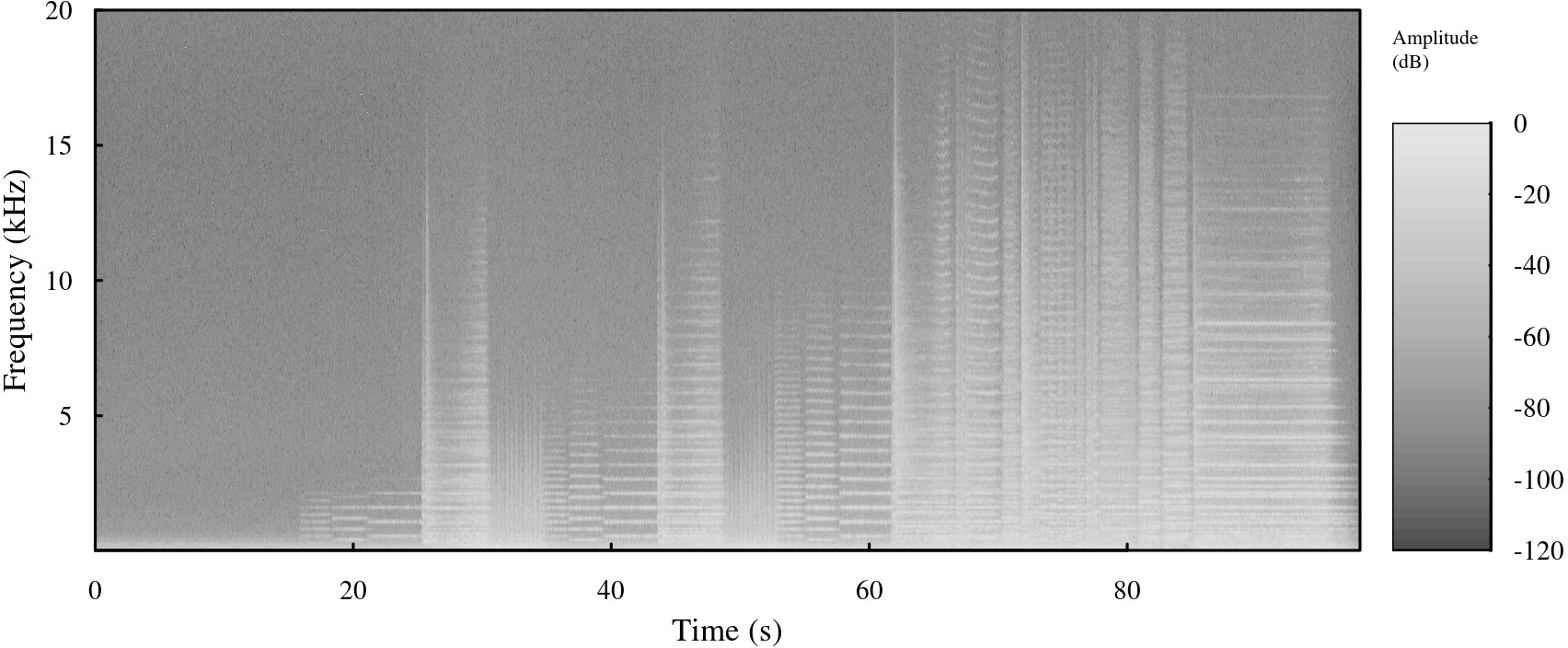}
\caption{Spectrogram of the right channel of the  opening fanfare  of ``Also
  Sprach Zarathustra'' (Op. 30) by Richard Strauss, performed by Vienna
  Philharmonic Orchestra conducted by Herbert von Karajan (Decca, 1959).  The
  power spectra values are expressed in dBFS scale and coded as colors ranging
  from black (low energy) to white (high energy).}

\label{fig:zarathustra}
\end{center}
\end{figure}
Figure \ref{fig:zarathustra} reports the spectrogram of a famous  fanfare
expressed in dBFS.  This is a particularly dynamic piece of sound. The orchestra
plays a soft opening followed by a series of transients  at full blast with
varying decay-time. There are several changes in the spectral distribution. At
particular time points there are peaks localized in several frequency bands, but
there is a continuum of energy spread between peaks. This shows how  major
variations are  characterized by a strong continuous component in the spectrum
that is also time-varying. In their  pioneering works \cite{Voss_Clarke_1975,
  Voss_Clarke_1978} found evidence that, at least for some musical instruments,
once the recorded  signal is passed trough band-pass filter with cutoffs set at
100Hz and 10KHz, the signal within the bandpass has a spectrum  that resembles
$1/f$--noise or similar fractal processes.  But the empirical evidence is based
on spectral methods acting as if the processes involved were stationary, whereas
this is often not true. Moreover, while most acoustic instruments produce most
of their energy between 100Hz and 10KHz, this is not true if one considers
complex ensembles. It is well known that for group of instruments playing
together, on average 50\% of the energy is produced in the range [20Hz,
300Hz]. Whether or not the $1/f$--noise hypothesis is true is yet to be
demonstrated, in this paper we give examples that show that the $1/f$--noise
hypothesis does not generally hold.

These observations are essential to motivate the following model for the PCM
samples.  Let $\{Y_t\}_{t \in\Z}$ be a sequence such that
\begin{equation}\label{eq:Yt}
Y_t = s(t)+\varepsilon_t,
\end{equation}
under the following assumptions:
\begin{assumption}{\textbf{A1}}{}\label{ass:s}
The function $s(t)$ has a continuous second derivative. 
\end{assumption}
\begin{assumption}{\textbf{A2}}{}\label{ass:eps}
$\{\eps_t\}_{t \in \Z}$ is a strictly stationary and $\alpha$-mixing process
with mixing coefficients $\alpha(k)$, $\ex[\eps_t]=0$,  $\ex
\abs{\eps_t^2}^{2+\delta} < +\Inf$, and  $\sum_{k=1}^{+\Inf}
\alpha^{\delta/(2+\delta)}(k)<\Inf$ for some $\delta>0$.

\end{assumption}
(\ref{eq:Yt}) is by no means interpretable as a Tukey-kind  signal plus noise
decomposition. The observable (recorded) sound wave $Y_t$ deviates from $s(t)$
because of several factors:  (i) transient changes in acoustic energy; (ii)
several sources of noise injected in the recording path; (iii) non-harmonic
components. We call the process $\{\varepsilon_t\}_{t \in \Z}$ the ``{\sl
  stochastic sound wave}'' (SSW). The main difference with
\cite{Serra_Smith_1990} is that in their work $\varepsilon_t$ is a white noise,
and $s(t)$ is a sum of sinusoids. \cite{Serra_Smith_1990} are mainly interested
in the spectral structure, hence they use $s(t)$ to study the discrete part of
the spectrum. Moreover they are interested in simple sounds from single
instruments, hence they simply assume that $\varepsilon_t$ is a white
noise. Their assumptions are reasonable for simple sounds, but in general do not
hold for complex sounds from an ensemble of instruments.

\ref{ass:s} imposes a certain degree of smoothness for $s(\cdot)$. This is
because we want that the stochastic term absorbs transients while 
$s(\cdot)$ mainly models long-term  periodic (harmonic) components.   The
$\alpha$-mixing assumption allows to manage non linearity, departure from
Gaussianity, and  a certain slowness in the  decay of the dependence structure
of the SSW. Certainly the $\alpha$-mixing assumption  \ref{ass:eps} would not be
consistent with the long-memory nature of $1/f$--noise. However, the
$1/f$--noise features disappear once the long term harmonic components are
caught by $s(\cdot)$ (see discussion on the example of
Figure~\ref{fig_intheflash_residuals}).  Whereas \ref{ass:eps}  allows for
various stochastic structures, the restrictions on the moments and mixing
coefficients are needed for technical reasons. Nevertheless, the existence of
the fourth moment is not that strong in practice, because this would imply that
the SSW has finite power variations, which is something that has to hold
otherwise it would be impossible to record it.\\


\section{Estimation}\label{sec:estimation} 
The DR statistic proposed in this paper is estimated  based on the following subsampling algorithm. \\
\begin{algorithm}[!h]
\caption{blockwise smoothing}\label{algo}
\SetAlgoLined
\Input{PCM data, $b \in N, K \in N$}
\Output{DR statistic}
\BlankLine
Draw a random sample $\set{I_1, I_2, \ldots, I_K}$  from the set  $\set{1,2,\ldots, n-b+1};$\\
\For{$i=1,2,\ldots,K$}{
{optimally estimate $s(\cdot)$ on the subsample $Y^{\text{sub}}_i=\set{y_{I_i}, y_{I_{i}+1}, \ldots, y_{I_{i}+b-1}},$}\\
{compute the sampling variance of $\hat \eps$ over $Y^{\text{sub}}_i$}
}
Construct the empirical distribution of the variances of $\hat \eps$.\\
Compute the DR statistic based on empirical distribution of the variances of $\hat \eps.$
\end{algorithm}

In analogy with equation \eqref{eq:pcm_power}, the RMS power of the SSW is given by
the sampling standard deviation of $\ProcEps$. The main goal is to obtain an
estimate of the distribution of the variance of $\{\eps_t\}_{t \in \Z}$. 
Application of existing  methods would require nonparametric  estimation
of  $s(\cdot)$ on the entire sample. However, the sample size is typically of
the order of millions of observations. Moreover, since the smooth  component is
time-varying, one would  estimate $s(\cdot)$ by using kernel methods with local
window. It is clear that all this is computationally  unfeasible.  
Compared with the standard subsampling, the ``blockwise smoothing'' Algorithm 
gives  clear advantages: (i) randomization reduces the otherwise impossible 
large number of subsamples to be explored;  (ii) none of the computations is 
performed on the entire sample. In particular estimation of $s(\cdot)$ is performed
subsample-wise as in Proposition \ref{prop:2} and Corollary \ref{cor:1}; (iii)
estimation of $s(\cdot)$ on smaller blocks of observations allows to adopt a
global, rather than a local, bandwidth  approach. Points (ii) and (iii) are
crucial for the feasibility of the computing load.  The smoothing and the random
subsampling part of the procedure are disentangled in the next two Sections.

\subsection{Smoothing}\label{sec_smoothing}
This section treats the smoothing with respect to the entire sample. The theory
developed into this section is functional to the development of the
\textit{local} estimation of $s(\cdot)$ at subsample level. The latter will be
treated in Section \ref{sec_random_subsampling}. First notice that  without loss
of generality we can always rescale $t$ onto $(0,1)$ with equally spaced
values. Therefore, model (\ref{eq:Yt}) can be written as 
\begin{equation}
\label{eq:Ytbis}
Y_i = s(i/n)+\varepsilon_i, \qquad i=1,\ldots,n.
\end{equation}
Estimation of $s(t)$, $t\in(0,1)$, is performed based on the classical Priestley-Chao kernel estimator \citep[][]{priestley-chao-1972}
\begin{equation}\label{eq:ker}
\hat{s}(t) = \frac{1}{nh} \sum_{i=1}^n \mathcal{K}\left(\frac{t-i/n}{h}\right)y_i,
\end{equation}
under  the assumption 
\begin{assumption}{\textbf{A3}}{} \label{ass:K}
$\mathcal{K}(\cdot)$ is a density function symmetric about zero with compact
support. Moreover, $\mathcal{K}(\cdot)$ is Lipschitz continuous of some
order. The bandwidth $h \in H=[c_1n^{-1/5};c_2n^{-1/5}]$, where $c_1$ and $c_2$ are two positive constants
such that $c_1$ is arbitrarily small while $c_2$ is arbitrarily large.  
\end{assumption}
Without loss of generality, we will  use the Epanechnikov kernel for its well
known efficiency properties, but any other kernel function fulfilling
\ref{ass:K} is welcome.  
\cite{altman-1990} studied the kernel regression problem  when the error term
additive to the regression function exhibits serial correlation.  Furthermore in
the setup considered by \cite{altman-1990} the error term is a linear
process. The paper showed that when the stochastic term exhibits serial
correlation, standard bandwidth optimality theory no longer applies. The author
proposed an optimal bandwidth estimation which is based on a correction factor
that assumes that  the autocorrelation function  is known. Therefore Altman's
theory does not apply here for two reasons: (i) in this paper the
$\set{\eps_t}_{t \in \Z}$  is not restricted to the class of linear processes;
(ii) we do not assume that serial correlations are known. Let  
$\hat{\eps}_i = y_i - \hat{s}(i/n)$, and let us define the cross-validation function
\begin{equation}\label{eq:CV}
\text{CV}(h)= %
\quadre{1-\frac{1}{nh}\sum_{j=-M}^M \mathcal{K}\left(\frac{j}{nh}\right)\hat{\rho}(j)} ^{-2} %
\frac{1}{n} \sum_{i=1}^n \hat{\eps}_i^2  ;%
\end{equation}
where the first term  is the  correction factor \`{a} la Altman with the
difference that  it  depends on the estimated autocorrelations of $\ProcEps$ up
to $M$th order. We show that the modification does not affect consistency at the
optimal rate. The number of lags into the correction factor depends both on $n$
and $h$. Intuitively consistency of the bandwidth selector can only be achieved
if $M$ increases at a rate smaller than $nh$, and in fact we will need the
following technical requirement: 
\begin{assumption}{\textbf{A4}}{}\label{ass:M}
Whenever $n\to \Inf$; then $M \to \infty$ and $M=O(\sqrt{nh})$.
\end{assumption}
The previous condition makes clear the relative order of the two smoothing
parameters $M$ and $h$. The bandwidth is estimated by minimizing the
cross-validation function, that is  
\begin{equation*}
\hat h = {\text{argmin}}_{h \in H} \; \text{CV}(h).
\end{equation*}
Since $t$ is deterministic and equally spaced in $(0,1)$ and, using the approach as in \cite{altman-1990}, we can write the Mean Square Error (MSE) of $\hat s(t)$ as
\[
\mbox{MSE}(h;\hat s(t))=\frac{B_{\mathcal{K}}^2}{4}h^4[s''(t)]^2+\frac{V_{\mathcal{K}}}{nh}\sigma_{\varepsilon}^2(1+2S_{\rho}),
\]
where $s''(\cdot)$ is the second derivative of $s(\cdot)$,
$\sigma_{\varepsilon}^2=E(\varepsilon_t^2)$,
$S_{\rho}=\sum_{j=1}^{\infty}\rho(j)$, $B_{\mathcal{K}}=\int
u^2\mathcal{K}(u)du$ and $V_{\mathcal{K}}=\int\mathcal{K}^2(u)du$. Let
$\text{MISE}(h; \hat{s})$ be the Mean Integrated Square Error of
$\hat{s}(\cdot)$, that is  
\[
\text{MISE}(h;\hat s)=\int_0^1 \mbox{MSE}(h;\hat s(t))dt=\frac{B_{\mathcal{K}}^2}{4}h^4\int_0^1[s''(t)]^2dt+\frac{V_{\mathcal{K}}}{nh}\sigma_{\varepsilon}^2(1+2S_{\rho})
\]
and let $h^\star$ be the global minimizer of $\text{MISE}(h;
\hat{s})$. 
\begin{proposition}\label{prop:1}
Assume {\ref{ass:s}}, {\ref{ass:eps}}, {\ref{ass:K}} and {\ref{ass:M}}.
$\hat{h}/{h^\star} \convp 1$ as $n \to \Inf$. 
\end{proposition}
Proof of Proposition~\ref{prop:1} is given in the Appendix. It shows that  $\hat
h$ achieves  the optimal global bandwidth.  The previous result improves the
existing literature in several aspects. Previous works on kernel regression with
correlated errors all requires stronger assumptions on $\ProcEps$,  e.g.
linearity, Gaussianity, existence of high order moments and some stringent
technical conditions
\citep[see][]{altman-1990,altman_1993,Hart_1991,xia_li_2002,
  FranciscoFernandez_etal_2004}. None of the contributions in the existing
literature treats the choice of the smoothing parameters in the cross-validation
function, that is $M$. \cite{FranciscoFernandez_etal_2004} mentions its crucial
importance, but no clear indication on how to set it is given. \ref{ass:M}
improves upon this giving a clear indication of how  this tuning  has to be
automatically fixed in order to achieve optimality. In fact, Proposition
\ref{prop:1} suggests to take  $M=\lfloor\sqrt{n h}\rfloor$. Therefore the
smoothing step is completely data-driven. Notice that alternatively standard
cross-validation would require to fix two tuning parameters \citep[see Theorem
2.3 in][]{Hall-a-etal-1995}.  

\begin{figure}[!t]
\centering
\includegraphics[width=\textwidth]{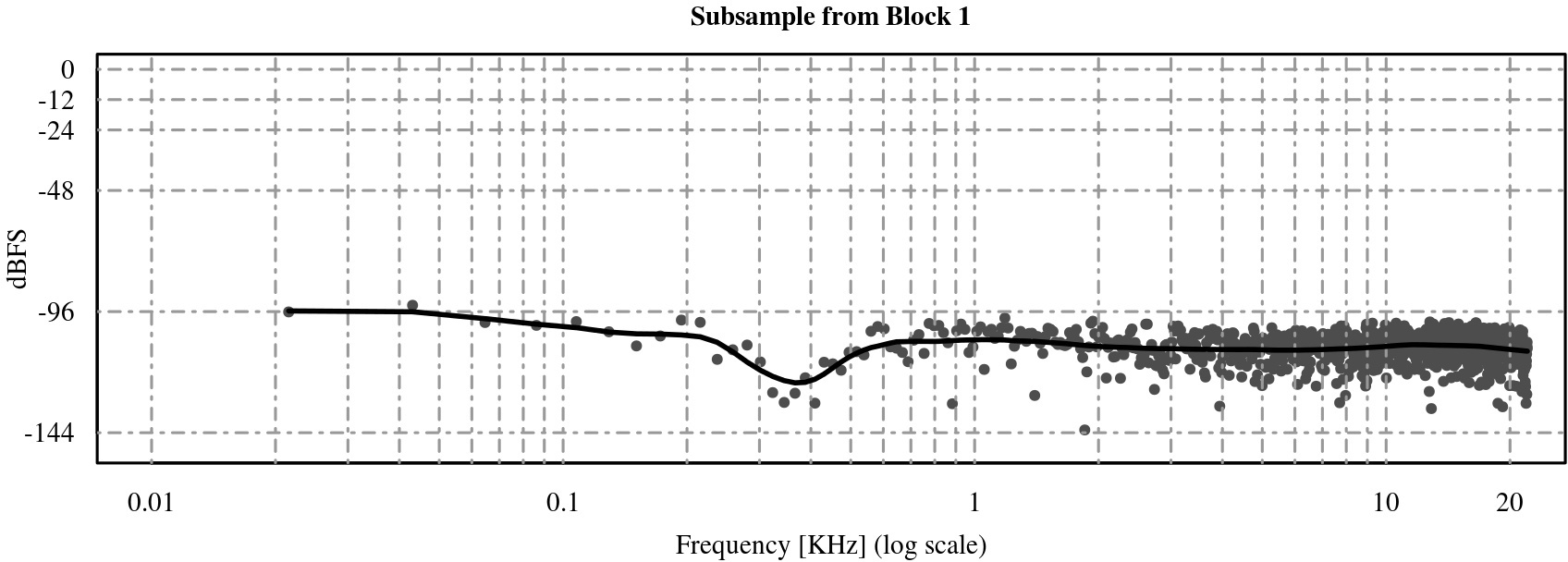}\\
\vspace{1em}
\includegraphics[width=\textwidth]{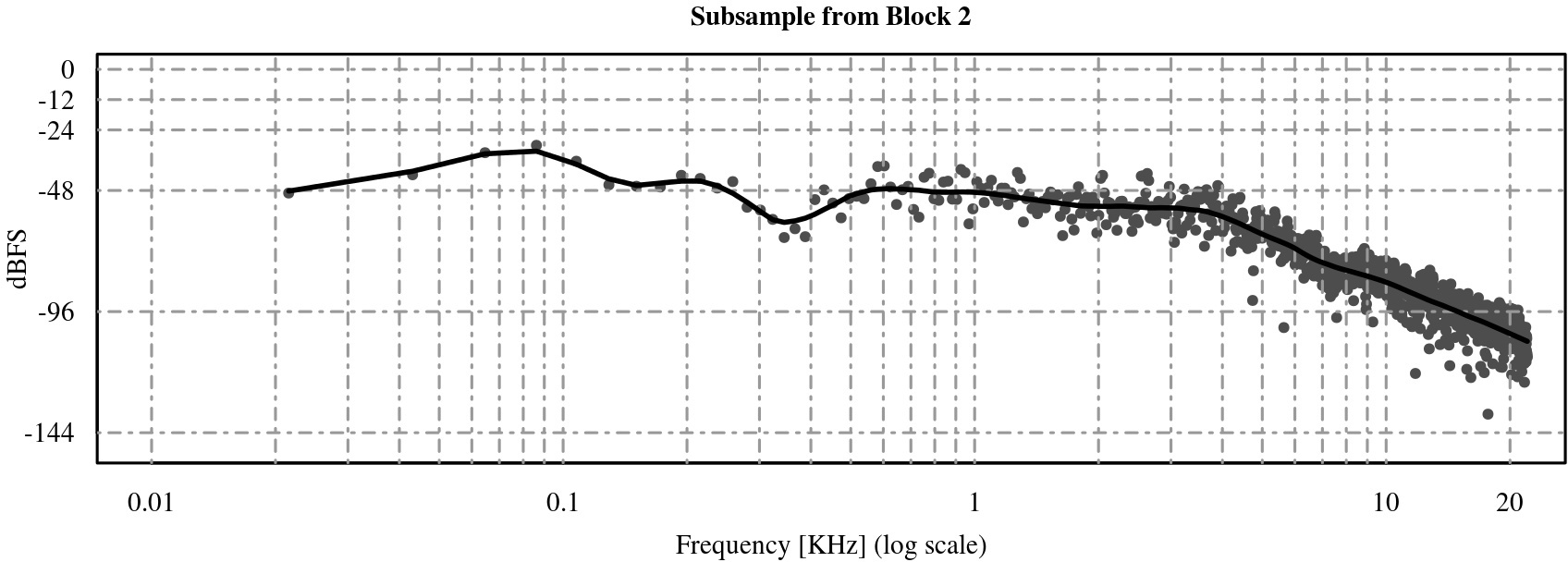}
\caption{Points represent windowed periodogram power spectral density estimates of the SSW  obtained in two subsamples of size 50ms extracted from the wave reported in Figure~\ref{fig:intheflash2039}. The solid line has been obtained by kernel smoothing.  The top plot refers to a subsample randomly chosen within ``Block 1'', while the bottom plot refers to a subsample randomly chosen within ``Block 2''.}
\label{fig_intheflash_residuals}
\end{figure}

\begin{figure}[!t]
\centering
\includegraphics[height=.3\textheight]{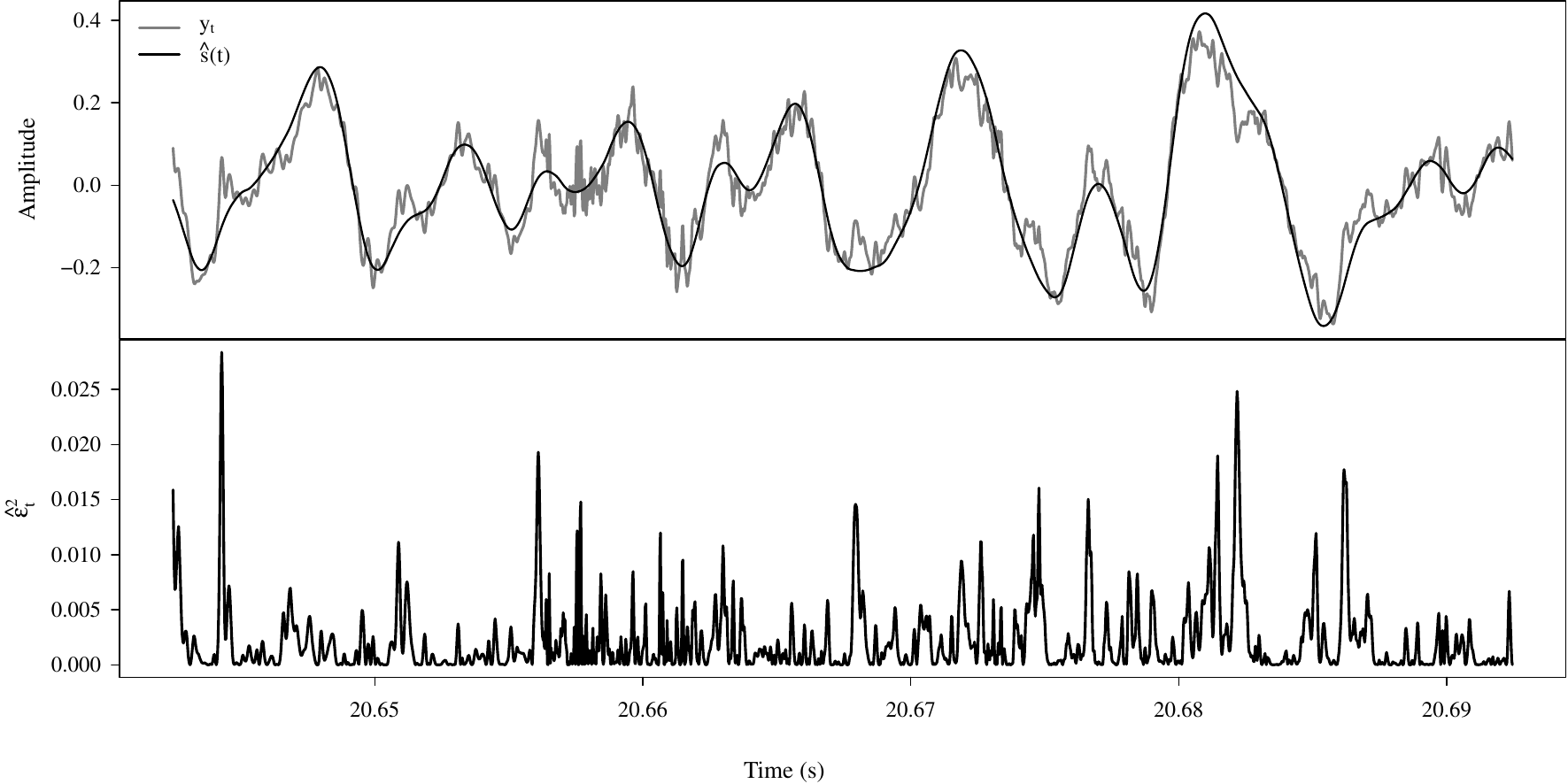}\\
\caption{Estimated signal and squared residuals for a subsamples of  50ms randomly 
extracted from the second half of the wave reported in Figure~\ref{fig:intheflash2039}.
The estimated optimal bandwidth is $\hat h=0.01521$.}
\label{fig_intheflesh_hat_s}
\end{figure}

In order to see how the behavior of $\ProcEps$ is time-varying, see
Figure~\ref{fig_intheflash_residuals}.  The SSW has been estimated based on
\eqref{eq:ker} and \eqref{eq:CV} on subsamples of length equal to 50ms. The
first subsample has been randomly chosen within Block 1 of
Figure~\ref{fig:intheflash2039}, while the second has been randomly chosen
within Block 2. Discrete-time Fourier transform measurements have been windowed
using Hanning window. Points in the plots correspond to spectral estimates at
FFT frequencies scaled to dBFS (log-scale).  It can be seen that the two
spectrum show dramatic differences. In the first one the energy spread by the
SSW is modest and near the shape and level of uncorrelated quantization
noise. On the other hand, the bottom spectrum shows a pattern that suggests that
correlations vanish at slow rate which is consistent with \ref{ass:eps}. The
steep linear shape of on log-log coordinates above 3KHz reminds us approximately
the shape of the $1/f$--noise spectrum, however below 3KHz the almost flat
behavior suggests a strong departure from the $1/f$--noise hypothesis.  All this
confirms the idea that there are music sequences where the SSW in \eqref{eq:Yt}
cannot be seen as the usual ``error term''. Tests proposed in
\cite{Berg_McMurry_Politis_2012} also lead to rejection of the linear hypothesis
for the SSW.  In the end, it is remarkable that these extremely diverse
stochastic structures coexist within just one second of music.
In Figure~\ref{fig_intheflesh_hat_s} we show the estimated $\hat s(\cdot)$ against
observed data (top panel), and  the corresponding $\hat \eps_t^2$ for a 50ms
subsample taken from the same song. The large spikes in the bottom panel
correspond to situations where the amplitude variations  increase unexpectedly
so that they are caught by the stochastic component of the model.

\subsection{Random Subsampling}\label{sec_random_subsampling}
Equation \eqref{eq:pcm_power} only takes into account sum of squares, this is
because theoretically PCM  are  always scaled to have zero mean. Notice that,
even though we assume that $\ProcEps$ has zero expectation, we define RMS based
on  variances taking into account the fact that quantization could introduce an
average offset in the PCM samples.  Let us introduce the following  quantities:   
\begin{equation}\label{eq:Vn}
V_n=\frac{1}{n-1}\sum_{i=1}^n\left(\eps_i-\bar{\eps} \right)^2, %
\qquad \text{with} \quad \bar{\eps}=\frac{1}{n} \sum_{i=1}^n \eps_i.
\end{equation}
The distribution of the RMS power of the SSW is given by the  distribution of $\sqrt{V_n}$.  By \ref{ass:eps} it can be shown that
$\sqrt{n} (V_n - \sigma^2_{\eps} ) \convd \text{Normal}(0,V),$ where
$\sigma^2_{\eps} =\ex[\eps_t^2]$ and $ V = \lim_{n\rightarrow\infty} n
\var[V_n]$. From now onward, $G(\cdot)$ will denote the distribution function of
a $\text{Normal}(0,V)$.  

Although  the sequence $\set{\eps_i}$ is not observable, one can approximate its
power distribution based on $\hat \eps_i=y_i-\hat{s}(i/n)$. Replacing $\eps_i$
with $\hat \eps_i$ in \eqref{eq:Vn} we  obtain: 
\begin{equation}\nonumber
\hat V_n=\frac{1}{n-1}\sum_{i=1}^n\left(\hat \eps_i-\bar{\hat \eps} \right)^2, %
\qquad \text{with} \quad \bar{\hat \eps}=\frac{1}{n} \sum_{i=1}^n \hat \eps_i.
\end{equation}
The distribution of $\hat V_n$ can now be used to approximate the distribution
of $V_n$.  One way to do this is to implement a subsampling scheme \`{a} la
\cite{politis-romano-wolf-2001}. That is, for all blocks of observations of
length $b$ (subsample size) one compute $\hat V_n$, in this case there would be
$n-b+1$ subsamples to explore. Then one hopes that the empirical distribution of
the $n-b+1$ subsample estimates of $\hat V_n$ agrees with the distribution of
$V_n$ when both $n$ and $b$ grow large enough at a certain relative speed. This
is essentially the subsampling scheme proposed by
\cite{politis-romano-wolf-2001} and \cite{Politis_Romano_2010}, but it is of
limited practical use here because the scale of $n$ is usually in millions, and
the computation of  $\hat s(\cdot)$ on the entire sample  would
require a huge computational effort that is  of order $O(n^2)$. This is because the Kernel estimation procedure requires a number o iterations of order $O(n)$ for each point of the support.  On a modern
computer\footnote{A computer equipped with an Intel i7 3.4GHz quad-core 
processor, 16GB of memory, and 64bit software.}  
an highly optimized software programmed in 
C language takes about 17 hours to compute $\hat s(\cdot)$ on a 3min song
($n=15,876,000$ samples), and about 38 hours for a 4.5min song
($n=23,814,000$). This figures are for a fixed global $h$. If we perform
cross-validation on a grid with 25 points (a reasonable choice), the estimate
becomes 17 days for a 3min song, and almost 40 days for a 4.5min song. And then
one needs to add the computing time needed for the subsampling steps which will
depend on $b$ and $K$.
The blockwise smoothing algorithm solves the problem by introducing a
variant to the classical subsampling scheme previously described. Namely instead of
estimating $s(\cdot)$ on the entire series, we estimate it on each subsample
separately, then we use the average estimated error computed block-wise instead
of $\bar{\hat \eps}$ computed on the whole sample. Moreover a blockwise kernel
estimate of $s(\cdot)$ allows to work with the simpler global bandwidth instead
of the more complex local bandwidth without loosing too much in the smoothing
step. The blockwise smoothing algorithm proposed in this paper (with the default choice of $b$ and $K$
discussed afterwards) computes the proposed DR statistic
in about 30s on the same computer for a 3min song, this saves us 17 days of
computing.\\
 

Let $\hat s_b(\cdot)$ be the estimator of $s(\cdot)$ on a subsample of length $b$, that is 
\begin{equation}\label{eq_sb}
\hat{s}_b(t) = \frac{1}{bh} \sum_{i=1}^b \mathcal{K}\left(\frac{t-i/b}{h}\right)y_i.
\end{equation}
At a given time point $t$ we consider a block of observations  of length $b$ and
we consider the following statistics  
\begin{equation}\nonumber
V_{n,b,t}=\frac{1}{b-1}\sum_{i=t}^{t+b-1} (\eps_i- \bar{\eps}_{b,t})^2, %
\qquad \text{and} \qquad
\hat{V}_{n,b,t}=\frac{1}{b-1}\sum_{i=t}^{t+b-1} (\hat{\eps}_i-\bar{\hat{\eps}}_{b,t})^2,
\end{equation}
with $\bar{{\eps}}_{b,t}=b^{-1}\sum_{i=t}^{t+b-1} {\eps}_i$ and
$\bar{\hat{\eps}}_{b,t}=b^{-1}\sum_{i=t}^{t+b-1} \hat{\varepsilon}_i$.  Note
that in $\hat V_{n,b,t}$ we can consider  either $\hat{\varepsilon}_{i}=y_i-\hat
s(i/n)$ or $\hat{\varepsilon}_{i}=y_i-\hat s_b((i-t+1)/b)$, $i=t,\ldots,t+b-1$.
Of course, the bandwidth, $h$, depends on $n$ or $b$. We do not report this
symbol because it will be  clear from the context. The empirical distribution
functions of $V_{n,b,t}$ and $\hat{V}_{n,b,t}$ will  be computed as  
\begin{eqnarray}
G_{n,b}(x)&=&\frac{1}{n-b+1}\sum_{t=1}^{n-b+1}\uno\set{\sqrt{b}\left(V_{n,b,t}-V_n\right)\le
x}, \nonumber\\
\hat{G}_{n,b}(x)&=&\frac{1}{n-b+1}\sum_{t=1}^{n-b+1} \uno\set{\sqrt{b} (\hat{V}_{n,b,t}-V_n)\le x} ; \nonumber
\end{eqnarray}
where $\uno\set{A}$ denotes the usual indicator function of the set $A$. 
Furthermore,  the quantiles of the subsampling distribution also converges to
the quantities of interest, that is those of $V_n$. This is a consequence of the
fact that  $\sqrt{n}V_n$ converges weakly to a Normal distribution, let it be
$F$. Let  define the empirical distributions: 
\begin{eqnarray}
F_{n,b}(x)&=&\frac{1}{n-b+1}\sum_{t=1}^{n-b+1}\uno \set{ \sqrt{b}V_{n,b,t}\le
x},\nonumber\\
\hat{F}_{n,b}(x)&=&\frac{1}{n-b+1}\sum_{t=1}^{n-b+1} \uno \set{\sqrt{b}\hat{V}_{n,b,t}\le
x}. \nonumber
\end{eqnarray}
For $ \gamma \in (0,1)$ the quantities  $q(\gamma)$, $q_{n,b}(\gamma)$ and
$\hat{q}_{n,b}(\gamma)$ denote respectively the $\gamma$-quantiles with respect
to the distributions $F$, $F_{n,b}$ and $\hat{F}_{n,b}$.   We adopt  the usual
definition that  $q(\gamma)=\inf\set{x: F(x)\ge \gamma}$.  
However exploring all subsamples makes the procedure still computationally
heavy. A second variant is to reduce the number of subsamples by introducing a
random block selection. Let $I_i$, $i=1,\ldots K$ be random variables indicating
the initial point of every block of length $b$. We draw the sequence
$\set{I_i}_{i=1}^K$, with or without replacement, from the set
$I=\{1,2,\ldots,n-b+1\}$. The empirical distribution function of the subsampling
variances of $\eps_{t}$ over the random blocks will be: 
$$
\tilde{G}_{n,b}(x)=\frac{1}{K}\sum_{i=1}^{K} \uno \set{\sqrt{b} \tonde{ \hat{V}_{n,b,I_i}-V_{n}} \leq x },
$$
and the next results states the consistency of $\tilde{G}$ in  approximating $G$.
\begin{proposition}\label{prop:2} %
Assume {\ref{ass:s}}, {\ref{ass:eps}}, {\ref{ass:K}} and {\ref{ass:M}}.  Let
$\hat s_b(t)$ be the estimator of $s(t)$ on a subsample of length $b$. If $K
\rightarrow \Inf$, $b/n\rightarrow 0$, $b\rightarrow \Inf$  then
$\sup_x\left|\tilde{G}_{n,b}(x)-G(x)\right| \convp 0$ when $n\rightarrow \Inf$. 
\end{proposition}
We can also establish consistency for the quantiles based on $\set{
  \hat{V}_{n,b,I_i} }_{i=1}^{K}$. Let define the distribution function 
$$
\tilde{F}_{n,b}(x) = \frac{1}{K} \sum_{t=1}^{K} \uno \set{ \sqrt{b}\hat{V}_{n,b,I_t}\le x},
$$
and let  $\tilde{q}_{n,b}(\gamma)$ be the $\gamma$-quantile with respect to  $\tilde{F}$.
\begin{corollary}\label{cor:1} %
Assume {\ref{ass:s}}, {\ref{ass:eps}}, {\ref{ass:K}} and {\ref{ass:M}}.  Let
$\hat s_b(t)$ be the estimator of $s(t)$ on a subsample of length $b$. If $K
\rightarrow \Inf$, $b/n\rightarrow 0$, $b\rightarrow \Inf$  then
$\tilde{q}_{n,b}(\gamma) \convp q(\gamma)$ when $n\rightarrow \Inf$. 
\end{corollary}
Proposition \ref{prop:2} and Corollary \ref{cor:1} are novel in two
directions. First, the two statements are based on  $\hat{\varepsilon}_i$ rather
than observed $\varepsilon_i$ as in standard subsampling. Second, we replace
$s(\cdot)$ by $s_b(\cdot)$ allowing for local smoothing without using
local-windowing on the entire sample.\\ 
Hence the subsampling procedure proposed here  consistently estimates the
distribution of $V_n$ and its quantiles. The key tuning constant of the
procedure is $b$. One can estimate an optimal $b$,  but again we have to accept
that the astronomically large  nature of $n$  would take  the whole estimation
time infeasible. Moreover, for the particular problem at hand, there are subject
matter considerations that can effectively drive the choice of  $b$. For music
signals dynamic variations are usually investigated on time intervals ranging
from 35ms to 125ms (these are metering ballistics established with the
IEC61672-1 protocol). Longer time horizons up to 1s are also used, but these are
usually considered for long-term noise pollution monitoring. In professional
audio software,  50ms is usually the default starting value. Therefore, we
suggest to start from $b=2205$ as the ``50ms--default''  for signals recorded at
the standard 44.1KHz sampling rate.

  
\section{Dynamic range statistic}\label{sec:dynam-range-stat} 
The random nature of $\eps_t$ allows us to use statistical theory to estimate
its distribution. If the SSW catches transient energy variations, then its
distribution will highlight important information about the dynamic.  The square
root of $\hat V_{n,b,I_i}$ is a consistent estimate of the RMS power of $\eps_t$
over the block starting from $t=I_i$. The loudness of the $\eps_t$ component
over each block can be measured on the dBFS scale taking 
$-10\log_{10}\hat V_{n,b,I_i} $.  In analogy with DRs we can define a DR
measure based on the subsampling distribution of $\hat V_{n,b,I_i}$.  We define
the DR measure blockwise as $\text{DR}_{n,b,I_i} = -10\log_{10}\hat V_{n,b,I_i}$. For a
sound wave scaled onto the interval [-1,1] this is actually a measure of DR of
$\eps_t$ because it tells us how much the SSW is below the maximum attainable
instantaneous power. We propose a DR statistic, the  
``{\sl Median Stochastic DR}'', defined as the median of the
subsampling distribution of $\text{DR}_{n,b,\cdot}$:
\[ 
\text{MeSDR} = \text{med}_K \{\text{DR}_{n,b,I_i}; \;\; i=1,2\ldots,K\},
\] 
where $\text{med}_K\{\cdot\}$ denotes the empirical median over a set of $K$
observations.  The MeSDR is a consistent estimator of
$-10\log_{10}\sigma_{\varepsilon}^2$, where
$\sigma^2_{\varepsilon}=\mbox{E}(V_n)$ is a parameter of the process $\ProcEps$.
In order to see this note that: %
(i) in the limit $\sigma^2_\varepsilon$  is the center of a symmetric distribution (see  Section\ref{sec_random_subsampling}); 
(ii) by Corollary \ref{cor:1} the median of the subsampling quantites $\hat{V}_{n,b,\cdot}$ is consistent for the median of $V_n$; 
(iii) the $\log(\cdot)$ function is continuous and strictly monotonic, hence the median is invariant to such a transformation. 
Of course in the asymptotic regime with probability one both the sample mean, and the sample median would give the same consistent estimate for $\sigma_{\varepsilon}^2$, but in finite samples there are reasons to prefer the
median. 
A DR statistic is a measure of spread that compares the location
of the power distribution with its peak.  Existing DR measures are based on mean power (e.g. the DRs in \eqref{eq:AvgDR}), but in finite samples the estimated mean is pushed toward the extreme peaks so that the spread of location-vs-peak may not be  well represented. Since the median is less influenced by the peaks, it is more appropriate to represent the spread. This is better understood based on results in Section \ref{sec:empirical_evidence_montecarlo}. \\

The numerical interpretation of MeSDR is straightforward. Note that for waves
scaled onto [-1,1] it is easy to see that our statistic is expressed in dBFS. If
the wave is not scaled onto [-1,1], it suffices to add $20\log_{10}$( maximum
absolute observed sample), and this will correct for the existence of
headroom. Suppose MeSDR=20, this means that 50\% of the stochastic sound power
is at least 10dBFS below the maximum instantaneous power. Large values of
MeSDR indicate large dynamic swings.  Furthermore, it can be argued that the SSW
not only catches transients and non-periodic smooth components. In fact, it's
likely that it fits noise, mainly quantization noise. This is certainly true,
but quantization noise operates at extremely low levels and its power is
constant over time. Moreover, since it is likely that digital operations
producing DR compression (compressors, limiters, equalizers, etc.) increase the
quantization noise (in theory this is not serially correlated), this would
reflect in a decrease of MeSDR, so we should be able to detect DR compression
better than classical DRs-like measures. \\

  
\section{Simulation experiment}\label{sec:empirical_evidence_montecarlo}

\begin{figure}[!t]
\centering
\includegraphics[width=\textwidth]{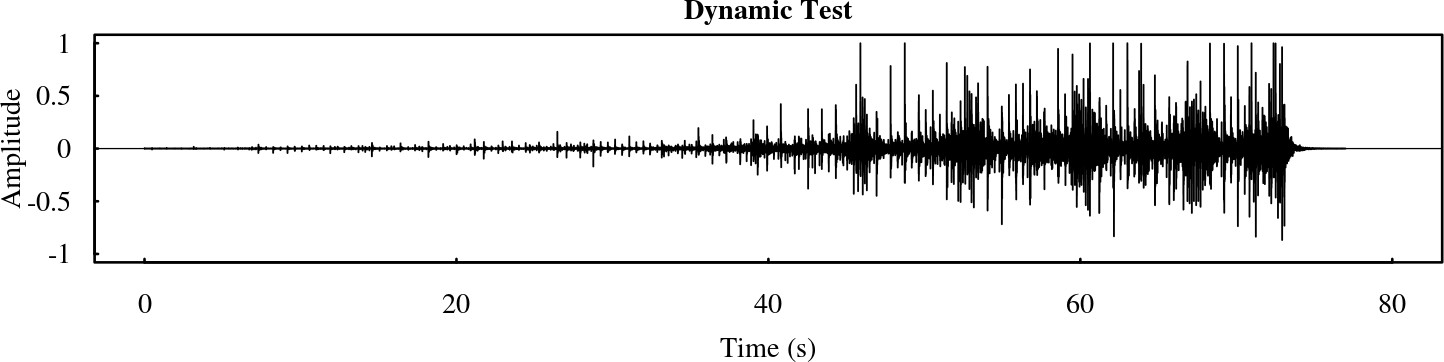}\\
\vspace{1em}
\includegraphics[width=\textwidth]{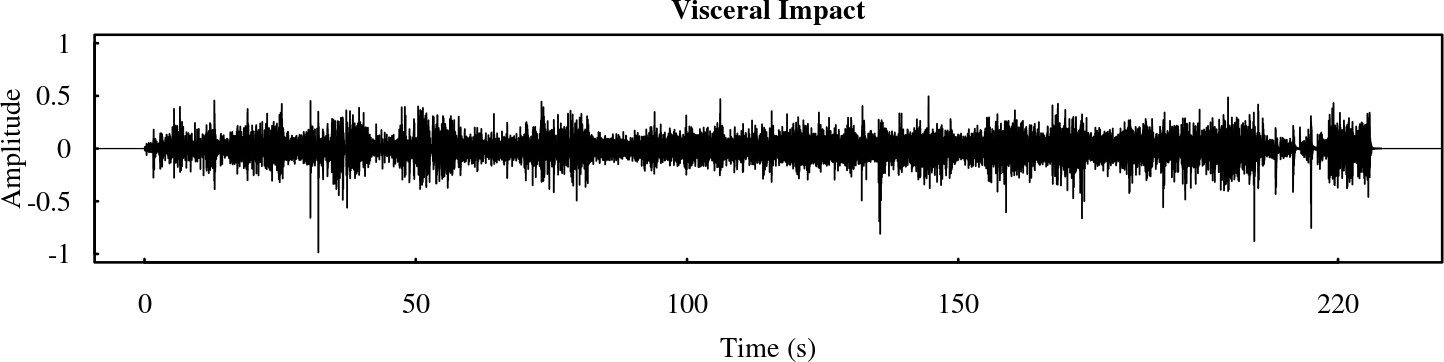}
\caption{Waveforms of the left channel of the songs titled  ``{\sl Dynamic Test}'' (track no.29), and ``{\sl Visceral Impact}'' (track no.17) from the audiophile CD  ``{\sl Ultimate Demonstration Disc: Chesky Records' Guide to Critical Listening}'' (Chesky Record, catalog  no. UD95).}%
\label{fig:chesky_tracks}
\end{figure}
A common way of assessing a statistical procedure is to simulate data from a
certain known stochastic process fixed as the reference truth, and then compute
Monte Carlo expectations of bias and efficiency measures. The problem here is
that  writing down a stochastic model capable of reproducing the features of
real-world music signal is too complex. Instead of simulating such a signal we
assess our methods based on simulated perturbations on  real data. We considered
two well recorded songs  and we added various degrees of dynamic compression to
assess  whether our measure is able to highlight dynamic differences. A good
method for estimating a measure of DR should consistently measure the loss of DR
introduced by compressing the dynamic. In order to achieve a fair comparison we
need songs  on which little amount of digital processing has been
applied. Chesky Records is a small  label specialized in audiophile recordings,
their ``{\sl Ultimate Demonstration Disc: Chesky Records' Guide to Critical
  Listening}'' (catalog number UD95), is almost a standard among audiophiles as
test source for various aspects of hifi reproduction.  We consider the left
channel   of tracks no.29, called ``{\sl Dynamic Test}'', and track no.17 called
``{\sl Visceral Impact}''. Both waveforms are reported in Figure
\ref{fig:chesky_tracks}. The ``Dynamic Test'' consists of a drum recorded near
field played with an increasing level. Its sound power is so huge that  a  voice
message warns against play backs at deliberately  high volumes,  which in fact
could cause equipment and hearing damages. Most audiophiles subjectively
consider this track as one of the most illustrious example of dynamic
recording. The track is roughly one  minute long.  The ``{\sl Visceral Impact}''
is actually the song ``{\sl Sweet Giorgia Brown}'' by Monty Alexander and
elsewhere published in the  Chesky catalog. The song has an energetic groove
from the beginning to the end and it's about three minutes long. Differently
from the previous track, that has increasing level of dynamic, this song has a
uniform path. This can be clearly seen in   Figure  \ref{fig:chesky_tracks}.\\ 

We removed initial and final silence from both tracks, and the final length (in
sample units) for ``{\sl Dynamic Test}'' is $n=2,646,000$, while $n=7,938,000$
for the second song.  We then applied compression on both waves. A dynamic
compressor is a function that whenever the original signal exceeds  a given
power (threshold parameter), the power of the output is scaled down  by a
certain factor (compression ratio parameter).     
With a threshold of -12dBFS and a compression ratio of 1.5, whenever the signal
power is above -12dBFS, the compressor reduces the signal level to $2/3$ so that
the input power is 1.5${\times}$(output power).  All this  has been performed
using SoX, an high quality audio digital processing software, with all other
tuning parameters set at default values. For both threshold levels equal to
-12dBFS and -24dBFS we applied on each song compression ratios equal to
$\{1.5,2,2.5,3,3.5,4,4.5,5\}$. There are a total of 16  compressed versions of
the original wave for each song. Hence the total number of tracks involved in
the simulation experiment is 34. Even though the random subsampling makes the
computational effort feasible, 34 cases still require a considerable amount of
computations.  The subsampling algorithm has been run with  $K=500$ for both
songs, while $b=2205$ (which means 50ms) for ``{\sl Dynamic Test}'', and
$b=3528$ (which means 80ms) for ``{\sl Visceral Impact}''. A larger  $b$ for
larger $n$ obeys  the theoretical requirements that $b/n \conv  0 $ as $n$ and
$b$ grow to $\infty$. The constant $M$ is fixed according to Proposition
\ref{prop:1}, i.e. $M= \lfloor \sqrt{n \hat h} \rfloor$. Stability analysis has
been conducted changing these parameters. In particular, we  tried several
values of $b$ for both tracks, but results did not change overall. A larger
value of $K$ also had almost no impact on the final results, however larger $K$
increases the computational load considerably. In a comparison like this, one
can choose to fix the seeds for all cases so that statistics are computed  over
the same subsamples in all cases.  However this would not allow to assess the
stability of the procedure against subsampling induced variance. The results
presented here are obtained with different seeds for each case, but fixed seed
has been tested and it  did not change the main results. Moreover, we estimate
$s(\cdot)$ in $(h,1-h)$ to avoid the well known issue of the boundary effect for
the kernel estimator. 
\begin{figure}[!t]
\centering
\includegraphics[width=.49\textwidth]{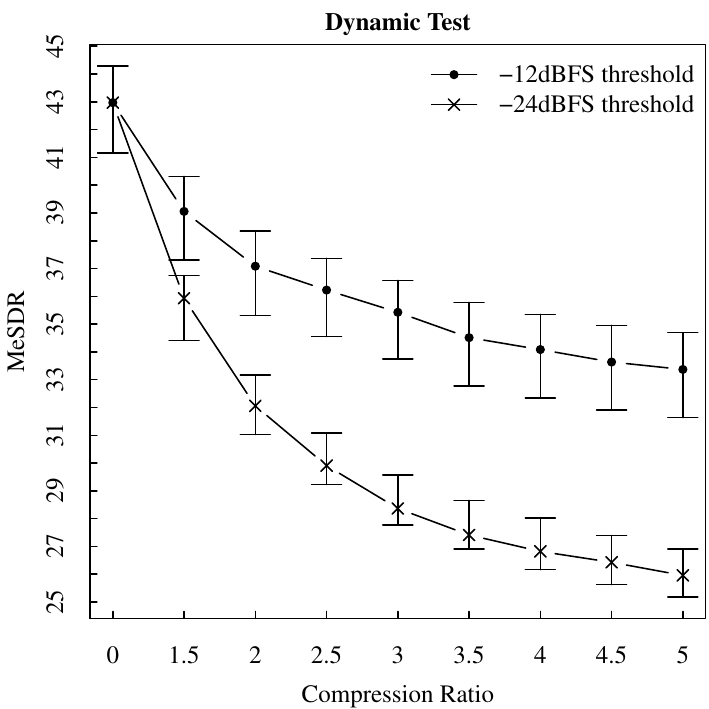}
\hfill
\includegraphics[width=.49\textwidth]{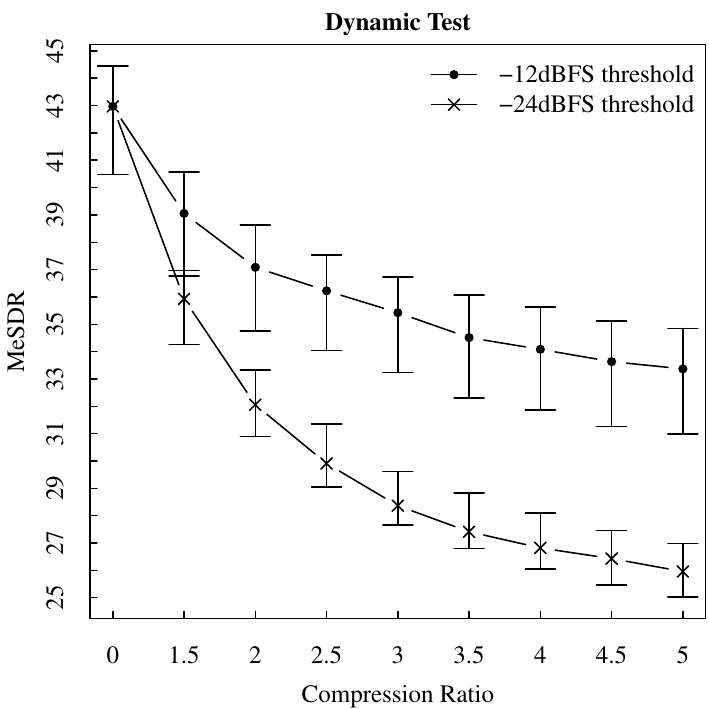}
\caption{MeSDR statistics for the song ``Dynamic Test''. The left plot reports
  90\%-confidence bands, while the right plot  reports 95\%-confidence bands.}%
\label{fig_msdr_dynamic}
\end{figure}
\begin{figure}[!t]
\centering
\includegraphics[width=.49\textwidth]{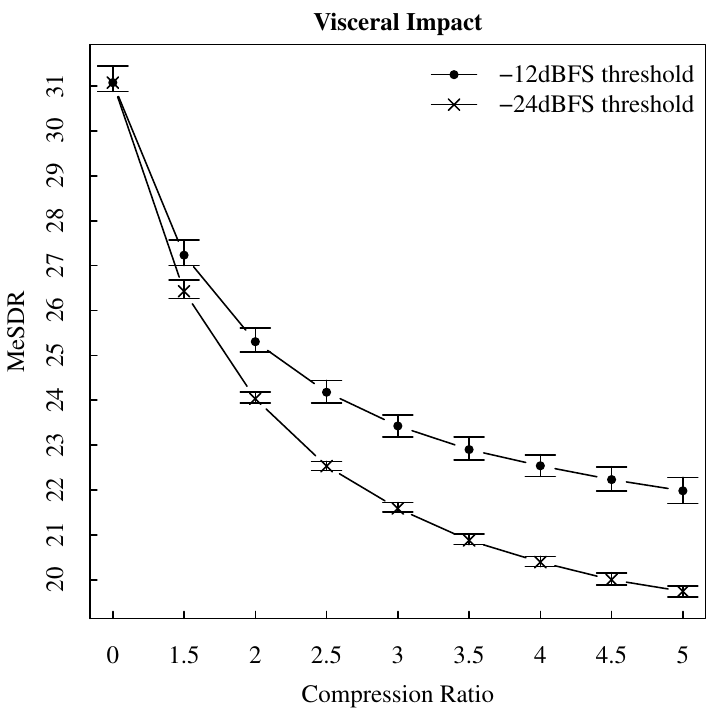}
\hfill
\includegraphics[width=.49\textwidth]{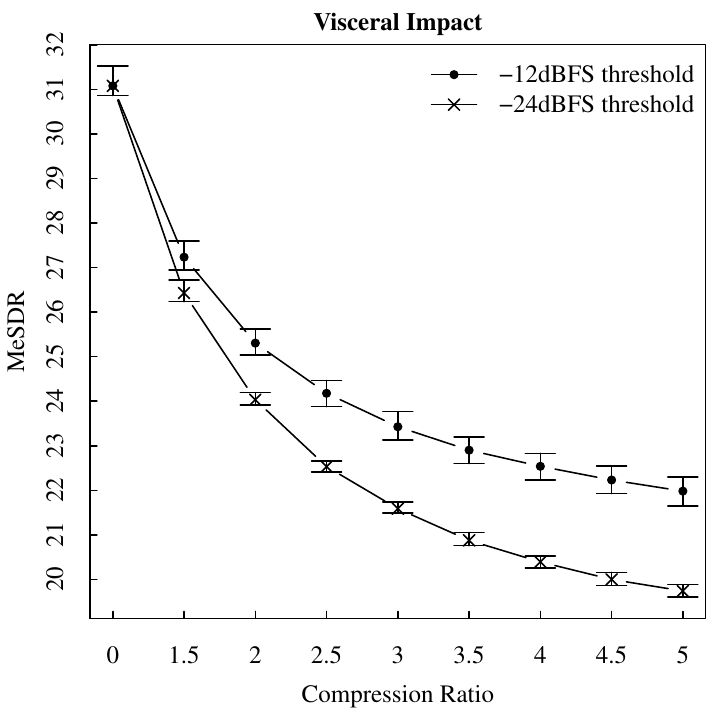}
\caption{MeSDR statistics for the song ``Visceral Impact''. The left  plot
  reports 90\%-confidence bands, while the right plot reports 95\%-confidence
  bands.}%
\label{fig_msdr_visceral}
\end{figure}

The results for all the cases are summarized in Figures \ref{fig_msdr_dynamic}
and \ref{fig_msdr_visceral} where we report  MeSDR statistics with 90\%  and
95\%-confidence bands. The simulation experiment reveals an interesting
evidence. 
\begin{enumerate}

\item First notice that each of the curves in Figures \ref{fig_msdr_dynamic} and
  \ref{fig_msdr_visceral} well emulates the theoretical behaviour of DR vs
  compression ratios. In fact, if we had  an ideal input signal with constant
  unit RMS power, the  DR decreases at the speed of $\log_{10}
  (1/\text{compression ratio})$ for any threshold value. When the RMS power is
  not constant, a consistent DR statistics should still  behave  similar to
  $\log_{10} (1/\text{compression ratio})$ with a curvature that depends both on
  the threshold parameter and the amount of power above the threshold.  

\item At both threshold levels, for both songs the MeSDR does a remarkable
  discrimination between compression levels.  For a given positive compression
  level none confidence bands for the -12dBFS threshold overlaps with the
  confidence bands for the -24dBFS case. Over 32 cases, a single overlap happens
  in   Figure \ref{fig_msdr_dynamic} for compression ratio equal to 1.5 when
  consider the wider 95\%--confidence interval.  

\item For both songs the confidence intervals are larger for the  -12dBFS case,
  and on average the ``Dynamic Test'' reports longer intervals. This is expected
  because increasing  the threshold from -12dBFS to -24dBFS will increase the
  proportion of samples affected by compression so that the variations of MeSDR
  will be reduced. Moreover we also expect that if the dynamic of a song doesn't
  have a sort of  uniform path, as in the case of the ``Dynamic Test'', the
  variability of the MeSDR will be larger.  Summarizing, not only the level of
  MeSDR, but also the length of the bands (i.e. the uncertainty) revels
  important information on the DR.  

\item There is a smooth transition going from 90\% to 95\%--confidence
  intervals.This is an indication that the tails of the distribution of the
  MeSDR are well behaved.  
\end{enumerate}
The experiment above shows how MeSDR is able to detect consistently even small
differences in compression levels. Hence it can be used to effectively
discriminate between recording quality. 

\begin{figure}[!t]
\centering
\includegraphics[width=.49\textwidth]{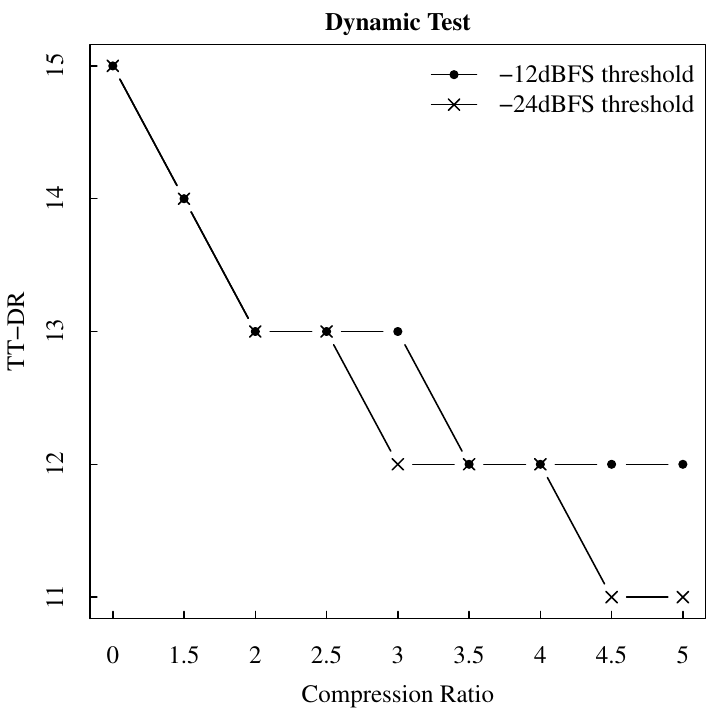}
\hfill
\includegraphics[width=.49\textwidth]{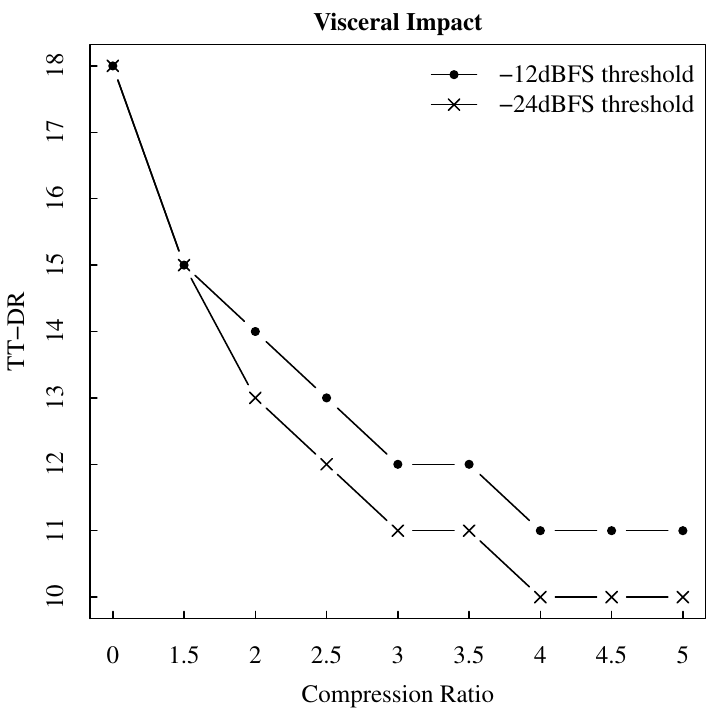}
\caption{TT-DR statistics for the song ``Dynamic Test'' (left) and ``Visceral Impact'' (right).}%
\label{fig_ttdr}
\end{figure}

In order to compare the results with  state of the art existing methods we
computed the TT-DR on the same data. The TT-DR is a popular measure of dynamic
range expressed in dBFS that  has gained a massive following. It is  is promoted
by the ``{\sl Pleasurize Music Foundation}'' (www.pleasurizemusic.com). TT-DR is
based on the sequential windowing of the  signal as for the DRs, but different
from this, it replaces the average power with the average of the power
measurements exceeding the 80\% quantile of the distribution. The idea behind
the TT-DR is that compression only affects the tail of the power
distribution. Results are reported in Figure \ref{fig_ttdr}. For the
uncompressed cases the TT-DR quantifies the dynamic of ``Dynamic Test'' as
larger than that of ``Visceral Impact'', although it is clear that the ``Dynamic
Test'' sounds more dynamic. This is because the TT-DR compares the mean of the 20\%
largest power measurements with the peak. The problem here is that the mean will
be pushed toward the peak if within the top 20\% power measurements there is
still an high degree of positive asymmetry.  The latter will happen particularly
in cases where the dynamic variations are as extreme as for the ``Dynamic
Test''. A major drawback of TT-DR is that overall the range of the statistic
across these cases does not exceeds 4dBFS for the ``Dynamic Test'', and 8dBFS
for the ``Visceral Impact'', which is an indication of the inability to capture
the DR concept consistently as the variations in dynamic levels here are
enormous.  For a given compression ratio the TT-DR does not always discriminate
between the -12dBFS and the -24dBFS thresholds, and overall the differences for
the two curves never exceeds 1dBFS. For a given threshold the TT-DR often  fails
to distinguish compression ratios, and again it is strange that this happens
more for the ``Dynamic Test''. The other disadvantage of the TT-DR is that the
sequential windowing (deterministic) does not allow to construct confidence
intervals and other inference tools. \\

\section{Real data application}\label{sec_real_dataset}
There are a number small record labels that gained success issuing remastered
versions of famous albums. Some of these reissues are now out of catalogue and
are traded at incredible prices on the second hand market. That means that music
lovers actually  value the recording quality.   On the other hand majors keep
issuing new remastered versions  promising miracles, they often  claim the use
of new super technologies termed with spectacular names. But music lovers are
often  critical. Despite the marketing trend of mastering music with obscene
levels of dynamic compression to make records sounding louder, human ears
perceive dynamic compression better than it is thought.  In this section we
measure the DR of three different digital masterings   of the song {``In the
  Flesh?''} from ``{\sl The Wall}'' album by Pink Floyd. The album  is
considered one of the best rock recording of all times and {``In the Flesh?''}
is a champion in dynamic, especially in the beginning (as reported in Figure
~\ref{fig:intheflash2039}) and at the end. We analysed three different masters:
the MFSL by Mobile Fidelity Sound Lab (catalog  UDCD 2-537, issued in 1990);
EMI94 by EMI Records Ltd (catalog  8312432, issued in 1994); EMI01 produced by
EMI Music Distribution (catalog 679182, issued in 2001). There are much more
remasters of the album not considered here.  The EMI01 has been marketed as a
remaster  with superior sound obtained with state of the art technology. The
MFSL has been worked out by a company specialized in classic album
remasters. The first impression is that the MFSL sounds softer than the  EMI
versions. However, there are  Pink Floyd fans arguing  that the MFSL sounds more
dynamic, and overall is  better than anything else. Also the difference between
the EMI94 and EMI01 is often discussed on internet forums with fans arguing that
EMI01 did not improve upon EMI94 as advertised.\\ 

\begin{table}[!t]
  \centering
  \caption{MeSDR statistics for ``{\sl In the Flesh?}'' from ``{\sl The Wall''} album  by Pink Floyd. ``Lower'' and ``Upper'' columns are  limits of the confidence intervals based on the asymptotic approximation of the distribution of the empirical quantiles.}
  \label{tab:intheflash}
  \begin{tabular}{llccccc}
    \toprule
    Seed number & Version & MeSDR & \multicolumn{2}{c}{90\%--Interval} & \multicolumn{2}{c}{95\%--Interval}\\
    &         &       &  Lower & Upper & Lower & ~Upper \\ 
    \midrule
    Equal   & MFSL   & 29.77 & 29.32 & 30.29 & 29.22 & 30.35 \\ 
    & EMI94  & 25.96 & 25.65 & 26.48 & 25.38 & 26.61 \\ 
    & EMI01  & 26.00 & 25.66 & 26.52 & 25.40 & 26.67 \\ 
    \midrule
    Unequal & MFSL   & 29.55 & 29.31 & 29.94 & 29.27 & 30.22 \\ 
    & EMI94  & 25.63 & 25.36 & 26.11 & 25.20 & 26.18 \\ 
    & EMI01  & 25.81 & 25.52 & 26.18 & 25.48 & 26.34 \\ 
    \bottomrule
  \end{tabular}
\end{table}

We measured the MeSDR of the three tracks and compared the results. Since there
was a large correlation between the two channels, we only measured the channel
with the largest peak, that is the left one.  The three waves have been
time--aligned, and the initial and final silence has been trimmed.  The MeSDR
has been computed with a block length of $b=2205$ (that is 50ms), and
$K=500$. We computed the MeSDR both with equal and unequal seeds across the
three waves to test for subsampling induced variability when $K$ is on the low
side.  Results are summarized in Table \ref{tab:intheflash}.  First notice the
seed changes the MeSDR only slightly. Increasing the value of $K$ to 1000 would
make this difference even smaller. The second thing to notice is that  MeSDR
reports almost no difference between EMI94 and EMI01. In a way the figure
provided by the MeSDR is consistent with most subjective opinions that the MFSL,
while sounding softer, has more dynamic textures. We recall here that a 3dBFS
difference is about twice the dynamic in terms of power. Table
\ref{tab:intheflash} suggests  that the MFSL remaster is about  4dBFS more
dynamic than competitors, this is a huge difference.\\  

A further question is whether there are statistically significant differences
between the sound power distributions of the tracks. It's obvious that the
descriptive nature of the approaches described in Section \ref{sec:tech} cannot
answer to such a question.  Our subsampling approach can indeed answer to this
kind of question by hypothesis testing.  If the mixing coefficients of
$\ProcEps$ go to zero at a proper rate, one can take the loudness measurements
$\set{ \text{DR}_{n,b,I_i}    \quad i=1,2,\ldots,K}$ as almost asymptotically
independent and  then  apply some standard tests. In order to assess the
differences in  MeSDR across the three masterings we performed the Mood's
median test \citep{Mood54}. The null hypothesis is that all masterings have the
same MeSDR against the alternative that at least one of them is different.  We
used $\text{DR}_{n,b,I_i}$ estimates from unequal seed calculations to reinforce 
the independence between the three groups. 
The resulting p-value is 8.95${\times}10^{-17}$, the latter confirms the subjective
perception that there are strong differences. Performing the Mood's test in pairs, the
p-value is extremely low when MFSL is compared against EMI94 or EMI01, while
the comparison between EMI94 and EMI01 gives a p-value=0.6579907. Although the
 median test here is a natural choice, \cite{FreidlinGastwirth2000} showed that
it has low power in many situations. The suggestion of \cite{FreidlinGastwirth2000} 
for two-sample problems is to use the Mann-Whitney tests for location shift.
Therefore, we applied the Mann-Whitney test for the null hypothesis that the DR distribution
of MFSL and EMI94 are not location-shifted, against the alternative that MFSL is
right shifted with respect to EMI94. The  resulting  p-value$<2.2\times
10^{-16}$ suggests to reject the null at any sensible significance  level, and this
confirms  that MFSL sounds more dynamic compared to  EMI94. Comparison between
MFSL and EMI01 leads to a similar result. We also tested the null hypothesis
that EMI94 and EMI01 DR levels  are not location-shifted, against the
alternative that there is some shift different from zero. The resulting
p-value=0.6797 suggests to not reject the null at any standard confidence
level. The latter confirms the figure suggested by MeSDR, that is, there is no
significant overall dynamic difference between the two masterings. \\

\section{Concluding Remarks}\label{sec:concl-final-remarks}
Starting from the DR problem we exploited a novel methodology to estimate the
variance distribution of a time series produced by a stochastic process additive
to a smooth function of time. The general set of assumptions on the error term
makes the proposed model flexible and general enough to be applied under various
situations not explored in this paper. The smoothing and the subsampling theory
is developed for fixed global bandwidth and fixed subsample size.  We
constructed a DR statistic that is based on the random term. This has two main
advantages: (i) it allows to draw conclusions based on inference; (ii) since
power variations  are about sharp changes in the  energy levels, it is likely
that these changes will affect the stochastic part of \eqref{eq:Yt} more than
the  smooth component. In a controlled experiment our DR statistic has been able
to highlight consistently dynamic range compressions. Moreover we provided an
example where the MeSDR statistic is able to reconstruct  differences perceived
subjectively on real music signals.

\appendix
\section*{\centering Appendix}\label{sec:proofs}

\paragraph{Proof of Proposition \ref{prop:1}.} First notice that \ref{ass:eps}--\ref{ass:K} in this paper imply that Assumptions A--E in \cite{altman-1990} are fulfilled. In particular \ref{ass:eps}  on   mixing coefficients of $\eps_t$ ensures that Assumptions E and D in \cite{altman-1990} are satisfied. 
Let %
$\hat{\gamma}(j)=\frac{1}{n}\sum_{t=1}^{n-j}\hat{\varepsilon}_t\hat{\varepsilon}_{t+j}$ %
be the estimator of the autocovariance $\gamma(j)$ with $j=0,1,\ldots$ and $r_n=\frac{1}{nh}+h^4=n^{-4/5}$ by \ref{ass:K}.  
First we note that by Markov inequality 
\begin{equation}\label{eq:proof_1_1}
\frac{1}{n}\sum_{i=1}^{n}\left(s(i/n)-\hat{s}(i/n)\right)^2 - \text{MISE}(h;\hat{s})  = o(r_n)
\end{equation} 
Let us now rearrange  $\hat{\gamma}(j)$ as 
\begin{eqnarray}
\label{eq:proof_1_2}
\hat{\gamma}(j) &=& %
\frac{1}{n}\sum_{i=1}^{n-j}\left(s(i/n)-\hat{s}(i/n)\right)\left(s((i+j)/n)-\hat{s}((i+j)/n)\right) + \nonumber \\  %
&+&  \frac{1}{n}\sum_{i=1}^{n-j}\left(s((i+j)/n)-\hat{s}((i+j)/n)\right)\varepsilon_{i} + \nonumber \\ %
&+&  \frac{1}{n}\sum_{i=1}^{n-j}\left(s(i/n)-\hat{s}(i/n)\right)\varepsilon_{i+j}+\frac{1}{n}\sum_{i=1}^{n-j}\varepsilon_i\varepsilon_{i+j} =I+II+III+IV 
\end{eqnarray}
By \eqref{eq:proof_1_1}  and Schwartz inequality it results that
$I=O_p\left(r_n\right)$. Now, consider term $III$ in \eqref{eq:proof_1_2}. Since
$\hat s(t)=s(t)\left(1+O_p\left(r_n^{1/2}\right)\right)$, it is sufficient to
investigate the behaviour of

\[
\frac{1}{n}\sum_{i=1}^{n-j}s(i/n)\varepsilon_{i+j}.
\]
By \ref{ass:s} and applying Chebishev inequality, it happens that
$III=O_p(n^{-1/2})$. Based on similar arguments one has that term
$II=O_p(n^{-1/2})$. Finally, $IV=\gamma(j)+O_p(n^{-1/2})$.  Then 
$\hat{\gamma}(j)=\gamma(j) + O_p(r_n) + O_p(n^{-1/2})+O_p(j/n)$, where the
$O_p(j/n)$ is due to the bias of $\hat{\gamma}(j)$. It follows that also
$\hat{\rho}(j) = \rho(j)  + O_p(r_n)+ O_p(n^{-1/2})+O_p(j/n)$. Since
$\mathcal{K}(\cdot)$ is bounded from above then one can write
\begin{align*}
\frac{1}{nh}\sum_{j=-M}^M \mathcal{K}\left(\frac{j}{nh}\right)\hat{\rho}(j) =& %
       \frac{1}{nh}\sum_{j=-M}^M \mathcal{K}\left(\frac{j}{nh}\right)\rho(j) + \\
{}     +& \frac{M}{nh}O_p(r_n)+\frac{M}{nh}O_p(n^{-1/2})+\frac{M^2}{n}O_p\left(\frac{1}{nh}\right);
\end{align*}
by \ref{ass:M} and $h=O(n^{-1/5})$, it holds true that
\begin{equation}
\label{eq:proof_1_3}
\frac{1}{nh}\sum_{j=-M}^M
\mathcal{K}\left(\frac{j}{nh}\right)\hat{\rho}(j)=\frac{1}{nh}\sum_{j=-M}^M
\mathcal{K}\left(\frac{j}{nh}\right)\rho(j)+o_p(r_n).
\end{equation}
So, taking the quantity in the expression (22) of
\citet{altman-1990}, we have to evaluate
\[
Q_1=\left|\frac{1}{nh}\sum_{j=-nh/2}^{nh/2}\mathcal{K}\left(\frac{j}{nh}\right)\rho(j)-\frac{1}{nh}\sum_{j=-M}^M
\mathcal{K}\left(\frac{j}{nh}\right)\hat{\rho}(j)\right|
\]
where we consider the nearest integer in place of $nh$. By \eqref{eq:proof_1_3},
it follows that 
\[
Q_1=\left|\frac{2}{nh}\sum_{j=M+1}^{nh/2}\mathcal{K}\left(\frac{j}{nh}\right)\rho(j)\right|+o_p(r_n).
\]
Since, by assumptions \ref{ass:eps}, \ref{ass:K} and 
\ref{ass:M},
\[
\frac{1}{nh}\sum_{j=M+1}^{nh/2}\mathcal{K}\left(\frac{j}{nh}\right)\rho(j)\sim\rho(nh)=o(r_n),
\]
then $Q_1=o_p(r_n)$. Notice that the first term in $Q_1$ is the analog of
expression (22) in \citet{altman-1990}. Therefore, $\text{CV}(h)$ can now be
written as
\begin{align*}
\text{CV}(h) &= %
\quadre{1-\frac{1}{nh}\sum_{j=-M}^M\mathcal{K}\left(\frac{j}{nh}\right)\hat{\rho}(j)}
^{-2}  \frac{1}{n} \sum_{i=1}^n \hat{\eps}_i^2\\
{} &=\quadre{1-\frac{1}{nh}\sum_{j=-nh/2}^{nh/2}\mathcal{K}\left(\frac{j}{nh}\right)\rho(j)} ^{-2} %
\frac{1}{n} \sum_{i=1}^n \hat{\eps}_i^2+o_p(r_n).
\end{align*}
Apply the classical bias correction and based on (14) in \citet{altman-1990}, we
have that
\[
\text{CV}(h)=\frac{1}{n}\sum_{t=1}^n\varepsilon_t^2+\text{MISE}(h;\hat{s})+
o_p(r_n),
\]
using the same arguments as in the proof of Theorem 1 in
\citet{chu-marron-1991}. Since $\text{MISE}(h;\hat s)=O(r_n)$, it follows that $\hat{h}$, the
minimizer of $\text{CV}(h)$, is equal to $h^\star$, the minimizer
of $\text{MISE}(h;\hat{s})$, asymptotically in probability. \qed \\

\noindent Before proving Proposition  \ref{prop:2}, we need a technical Lemma
that states the consistency of the  classical subsampling  (not random)  in the
case of the estimator $\hat s(t)$ on the entire sample.   
\begin{lemma}
\label{lemma1} Assume \ref{ass:s}, \ref{ass:eps}, \ref{ass:K} and \ref{ass:M}.  Let $\hat s(t)$ be the estimator of $s(t)$ computed on the entire sample (of length n). If $b/n\conv 0$ whenever $n\conv \infty$and $b\conv \infty$, then 
\begin{eqnarray}
(i)&\sup_x\left|\hat G_{n,b}(x)-G(x)\right|\stackrel{p}{\longrightarrow} 0 \nonumber \\
(ii) & \hat q_{n,b}(\gamma)\stackrel{p}{\longrightarrow}q(\gamma) \nonumber
\end{eqnarray}
where $\gamma\in(0,1)$ and $\hat G_{n,b}(x)$, $ G(x)$,
$\hat q_{n,b}(\gamma)$ and $q(\gamma)$ are defined in Section \ref{sec:estimation}. 
\end{lemma}
\begin{proof}
For the part \textit{(i)}, notice that under the assumption \ref{ass:eps}, the conditions of Theorem (4.1) in \cite{politis-romano-wolf-2001} hold. Let us denote  $r_n=\frac{1}{nh}+h^4$, which  is  $r_n=n^{-4/5}$ by Proposition (\ref{prop:1}). As in the proof of Proposition (\ref{prop:1}), we can write 
\begin{displaymath}
\frac{1}{b}\sum_{i=t}^{t+b-1}\left(\hat{\varepsilon}_i-\varepsilon_i\right)^2 = O_p(r_n) + O(1/b) \qquad \text{for all} \; t,
\end{displaymath}
since $\hat s(\cdot)$ is estimated on the entire sample of length $n$ and $\{\varepsilon_i\}$ is  a sequence of stationary random variables by \ref{ass:eps}. The term $O(1/b)$ is due to the error of the deterministic variable in $s(\cdot)$, when we consider the mean instead of the integral.  Note that we do not report this error ($O(1/n)$) in (\ref{eq:proof_1_1}) because the leading term is $r_n$ given that  $r_n^{-1}/n\rightarrow 0$ when $n\rightarrow\infty$.  In this case $b^{-1} \sum_{i=t}^{t+b-1}\left(\hat{\varepsilon}_i-\varepsilon_i\right)^2$ is an estimator of $\text{MISE}_I(h;\hat s)$ on a set $I\subset (0,1)$. Now   $\sqrt{b}(r_n+b^{-1}) \conv  0$ as $n \conv \infty$, and all this is sufficient to have that $\sqrt{b}\left(\hat{V}_{n,b,t}-V_{n,b,t}\right) \convp 0$, for all $t$. Then, we can conclude that $\sqrt{b}\left(\hat{V}_{n,b,t}-V_{n}\right)$ has the same asymptotic distribution as  $\sqrt{b}\left(V_{n,b,t}-V_n\right)$. Let us denote $Z_{1t}=\sqrt{b}\left(V_{n,b,t}-V_n\right)$ and $Z_{2t}=\sqrt{b}\left(\hat{V}_{n,b,t}-V_{n,b,t}\right)$,  hence 
\begin{displaymath}
\hat{G}_{n,b}(x) = \frac{1}{n-b+1}\sum_{t=1}^{n-b+1}   \uno \set{Z_{1t}+Z_{2t} \le x}.
\end{displaymath}
By the same arguments used for the proof of Slutsky theorem the previous equation can be written as
\begin{eqnarray}
\sup_x \abs{\hat{G}_{n,b}(x)-G(x)}   &\le&   \sup_x \abs{ G_{n,b}(x \pm \xi)-G(x)}  + \nonumber\\
                                    &+&     \frac{1}{n-b+1}\sum_{t=1}^{n-b+1} \uno \set{ |Z_{2t}| > \xi}  \nonumber
\end{eqnarray}
for any positive constant $\xi$. Since $G(x)$ is continuous at any $x$ (Normal distribution), it follows that  $\sup_x\abs{ G_{n,b}(x) -G(x)} \convp 0$ by Theorem (4.1) in \cite{politis-romano-wolf-2001}. Moreover, by \ref{ass:eps}  $Z_{2t} \convp 0$, for all $t$ and  thus 
\begin{displaymath}
\frac{1}{n-b+1}\sum_{t=1}^{n-b+1} \uno \set{ |Z_{2t}|> \xi} \convp 0, 
\end{displaymath}
for all $\xi>0$, which proves  the result.  

Finally, part \textit{(ii)} is straightforward following Theorem 5.1 of \cite{politis-romano-wolf-2001} and part \textit{(i)} of this Lemma.
\end{proof}

\paragraph{Proof of Proposition \ref{prop:2}.} Let $\mbox{P}^*(X)$ and $\mbox{E}^*(X)$ be the
conditional probability and the conditional expectation of a random variable $X$
with respect to a set $\chi = \set{Y_1,\ldots,Y_n}$. 
Here $\hat G_{n,b}(x)$ uses the estimator $\hat s_b(\cdot)$  on each subsample
of length $b$. Then, 
\begin{equation*}
\frac{1}{b}\sum_{i=t}^{t+b-1}\left(\hat \varepsilon_i-\varepsilon_i\right)^2=O_p\left(b^{-4/5}\right) \qquad \text{for all} \; t,
\end{equation*}
as in the proofs of Proposition \ref{prop:1} and Lemma \ref{lemma1} (i). Then, $\sqrt{b}b^{-4/5}\conv 0$. So, by Lemma 
\ref{lemma1} (i), it follows that
\[
\sup_x\left|\hat G_{n,b}(x)-G(x)\right|\convp 0,
\]
since $G(x)$ is continuous for all $x$. Let $Z_i(x)=\uno\set{\sqrt{b}\left(\hat
V_{n,b,i}-V_n\right)\le x}$  and $Z_i^*(x)=\uno\set{\sqrt{b}\left(\hat V_{n,b,I_i}-V_n\right)\le
x}$. $I_i$ is a random variable from $I=\set{1,2,\ldots,n-b+1}$.
Then, $\mbox{P}^*(Z_i^*(x)=Z_i(x))=\frac{1}{n-b+1}$ for all $i$ and each
$x$. Then, we can write $\tilde{G}_{n,b}(x)=\frac{1}{K}\sum_{i=1}^KZ_i^*(x)$ and
\begin{displaymath}
\mbox{E}^* \tonde{ \tilde{G}_{n,b}(x)} = \frac{1}{n-b+1}
\sum_{t=1}^{n-b+1}Z_i(x) = \hat{G}_{n,b}(x)\convp G(x).
\end{displaymath}
By Corollary 2.1 in \citet{politis-romano-1994} we have that
\[
\sup_x\left|\tilde G_{n,b}(x)-\hat G_{n,b}(x)\right|\conv 0\quad \mbox{almost sure}.
\]
Then, the result follows. \qed

\paragraph{Proof of Corollary \ref{cor:1}.}
It follows the proof of Lemma \ref{lemma1} (ii) by replacing  Lemma \ref{lemma1} (i)  with  Proposition \ref{prop:2}. \qed


\bibliographystyle{chicago}
\bibliography{0_REF}%

\begin{thebibliography}{}

\bibitem[\protect\citeauthoryear{Altman}{Altman}{1990}]{altman-1990}
Altman, N.~S. (1990).
\newblock Kernel smoothing of data with correlated errors.
\newblock {\em Journal of the American Statistical Association\/}~{\em
  85\/}(411), 749--759.

\bibitem[\protect\citeauthoryear{Altman}{Altman}{1993}]{altman_1993}
Altman, N.~S. (1993).
\newblock Estimating error correlation in nonparametric regression.
\newblock {\em Statistics \& probability letters\/}~{\em 18\/}(3), 213--218.

\bibitem[\protect\citeauthoryear{Ballou}{Ballou}{2005}]{ballou-2005}
Ballou, G. (2005).
\newblock {\em Handbook for Sound Engineers}.
\newblock Focal Press.

\bibitem[\protect\citeauthoryear{Berg, McMurry, and Politis}{Berg
  et~al.}{2012}]{Berg_McMurry_Politis_2012}
Berg, A., T.~McMurry, and D.~N. Politis (2012).
\newblock Testing time series linearity: traditional and bootstrap methods.
\newblock {\em Handbook of Statistics: Time Series Analysis: Methods and
  Applications\/}~{\em 30}, 27.

\bibitem[\protect\citeauthoryear{Boley, Lester, and Danner}{Boley
  et~al.}{2010}]{boley-etal-2010}
Boley, J., M.~Lester, and C.~Danner (2010).
\newblock Measuring dynamics: Comparing and contrasting algorithms for the
  computation of dynamic range.
\newblock In {\em Proceedings of the AES 129th Convention, San Francisco}.

\bibitem[\protect\citeauthoryear{Chu and Marron}{Chu and
  Marron}{1991}]{chu-marron-1991}
Chu, C.~K. and J.~S. Marron (1991).
\newblock Comparison of two bandwidth selectors with dependent errors.
\newblock {\em The Annals of Statistics\/}~{\em 19\/}(4), 1906--1918.

\bibitem[\protect\citeauthoryear{Francisco-Fern{\'a}ndez, Opsomer, and
  Vilar-Fern{\'a}ndez}{Francisco-Fern{\'a}ndez
  et~al.}{2004}]{FranciscoFernandez_etal_2004}
Francisco-Fern{\'a}ndez, M., J.~Opsomer, and J.~M. Vilar-Fern{\'a}ndez (2004).
\newblock Plug-in bandwidth selector for local polynomial regression estimator
  with correlated errors.
\newblock {\em Journal of Nonparametric Statistics\/}~{\em 16\/}(1-2),
  127--151.

\bibitem[\protect\citeauthoryear{Freidlin and Gastwirth}{Freidlin and
  Gastwirth}{2000}]{FreidlinGastwirth2000}
Freidlin, B. and J.~L. Gastwirth (2000).
\newblock Should the median test be retired from general use?
\newblock {\em The American Statistician\/}~{\em 54\/}(3), 161–164.

\bibitem[\protect\citeauthoryear{Hall, Lahiri, and Polzehl}{Hall
  et~al.}{1995}]{Hall-a-etal-1995}
Hall, P., S.~N. Lahiri, and J.~Polzehl (1995).
\newblock On bandwidth choice in nonparametric regression with both short- and
  long-range dependent errors.
\newblock {\em The Annals of Statistics\/}~{\em 23\/}(6), 1921--1936.

\bibitem[\protect\citeauthoryear{Hart}{Hart}{1991}]{Hart_1991}
Hart, J.~D. (1991).
\newblock Kernel regression estimation with time series errors.
\newblock {\em J. Roy. Statist. Soc. Ser. B\/}~{\em 53\/}(1), 173--187.

\bibitem[\protect\citeauthoryear{Katz}{Katz}{2007}]{katz-2007}
Katz, B. (2007).
\newblock {\em Mastering audio: the art and the science}.
\newblock Focal Press.

\bibitem[\protect\citeauthoryear{Mood}{Mood}{1954}]{Mood54}
Mood, A.~M. (1954).
\newblock On the asymptotic efficiency of certain nonparametric two-sample
  tests.
\newblock {\em The Annals of Mathematical Statistics\/}, 514--522.

\bibitem[\protect\citeauthoryear{Politis and Romano}{Politis and
  Romano}{1994}]{politis-romano-1994}
Politis, D.~N. and J.~P. Romano (1994).
\newblock Large sample confidence regions based on subsamples under minimal
  assumptions.
\newblock {\em The Annals of Statistics\/}~{\em 22\/}(4), 2031--2050.

\bibitem[\protect\citeauthoryear{Politis and Romano}{Politis and
  Romano}{2010}]{Politis_Romano_2010}
Politis, D.~N. and J.~P. Romano (2010).
\newblock {K}-sample subsampling in general spaces: The case of independent
  time series.
\newblock {\em Journal of Multivariate Analysis\/}~{\em 101\/}(2), 316--326.

\bibitem[\protect\citeauthoryear{Politis, Romano, and Wolf}{Politis
  et~al.}{2001}]{politis-romano-wolf-2001}
Politis, D.~N., J.~P. Romano, and M.~Wolf (2001).
\newblock {On the asymptotic theory of subsampling}.
\newblock {\em Statistica Sinica\/}~{\em 11\/}(4), 1105--1124.

\bibitem[\protect\citeauthoryear{Priestley and Chao}{Priestley and
  Chao}{1972}]{priestley-chao-1972}
Priestley, M. and M.~Chao (1972).
\newblock Non-parametric function fitting.
\newblock {\em Journal of the Royal Statistical Society. Series B
  (Methodological)\/}~{\em 34\/}(3), 385--392.

\bibitem[\protect\citeauthoryear{Risset and Mathews}{Risset and
  Mathews}{1969}]{risset-mathews-1969}
Risset, J.~C. and M.~V. Mathews (1969).
\newblock Analysis of musical-instrument tones.
\newblock {\em Physics today\/}~{\em 22}, 23.

\bibitem[\protect\citeauthoryear{Serra and Smith}{Serra and
  Smith}{1990}]{Serra_Smith_1990}
Serra, X. and J.~Smith (1990).
\newblock Spectral modeling synthesis: A sound analysis/synthesis system based
  on a deterministic plus stochastic decomposition.
\newblock {\em Computer Music Journal\/}~{\em 14\/}(4), 12.

\bibitem[\protect\citeauthoryear{Vickers}{Vickers}{2010}]{vickers-2010}
Vickers, E. (2010).
\newblock The loudness war: Background, speculation and recommendations.
\newblock In {\em Proceedings of the AES 129th Convention, San Francisco, CA},
  pp.\  4--7.

\bibitem[\protect\citeauthoryear{von Helmholtz}{von
  Helmholtz}{1885}]{vonhelmholtz-1885}
von Helmholtz, H. (1885).
\newblock {\em On the sensations of tone (English translation {AJ} Ellis)}.
\newblock New York: Dover.

\bibitem[\protect\citeauthoryear{Voss}{Voss}{1978}]{Voss_Clarke_1978}
Voss, R.~F. (1978).
\newblock ’’1/f noise’’ in music: Music from 1/f noise.
\newblock {\em The Journal of the Acoustical Society of America\/}~{\em
  63\/}(1), 258.

\bibitem[\protect\citeauthoryear{Voss and Clarke}{Voss and
  Clarke}{1975}]{Voss_Clarke_1975}
Voss, R.~F. and J.~Clarke (1975, Nov).
\newblock “1/fnoise” in music and speech.
\newblock {\em Nature\/}~{\em 258\/}(5533), 317–318.

\bibitem[\protect\citeauthoryear{Xia and Li}{Xia and Li}{2002}]{xia_li_2002}
Xia, Y. and W.~Li (2002).
\newblock Asymptotic behavior of bandwidth selected by the cross-validation
  method for local polynomial fitting.
\newblock {\em Journal of multivariate analysis\/}~{\em 83\/}(2), 265--287.

\end{thebibliography}
\end{document}